\pdfoutput=1
\RequirePackage{amssymb}
\documentclass[sigconf]{acmart}
\copyrightyear{2020}
\acmYear{2020}
\newcommand{\mvspace}[1]{\vspace{#1}}
\acmConference[]{}
\acmBooktitle{}
\AtBeginDocument{%
\fancyhead[LE]{}%
\fancyhead[RO]{}%
}
\setcopyright{none}
\renewcommand\footnotetextcopyrightpermission[1]{}
\usepackage{amsmath,amsthm}
\usepackage{xcolor}
\usepackage{xspace}
\usepackage{algorithm}
\usepackage[noend]{algpseudocode} %
\algnewcommand{\algorithmicassumption}{\textbf{Requirement:}}
\algnewcommand{\Assume}{\item[\algorithmicassumption]}
\algrenewcomment[1]{\hfill \(\triangleright\) \emph{\small #1}} %
\algnewcommand{\CommentLine}[1]{\(\triangleright\;\;\) \emph{\small #1} \(\;\;\triangleleft\)} %
\algnewcommand{\InlineIf}[2]{%
\algorithmicif\ #1\ \algorithmicthen\ #2}
\algnewcommand{\InlineIfElse}[3]{%
\algorithmicif\ #1\ \algorithmicthen\ #2\ \algorithmicelse\ #3}
\algnewcommand{\InlineWhile}[2]{%
\algorithmicwhile\ #1\ \algorithmicdo\ #2}
\algnewcommand{\InlineForAll}[2]{%
\algorithmicforall\ #1\ \algorithmicdo\ #2}

\usepackage{setspace}
\usepackage{color,svgcolor}
\usepackage{paralist}
\usepackage{tikz}
\usepackage{multirow}
\usetikzlibrary{matrix}
\usepackage[capitalize]{cleveref}
\usepackage{array} %
\newcolumntype{L}{>{$}p{1.9cm}<{$\hfil}} %
\makeatletter
\def\mathcolor#1#{\@mathcolor{#1}}
\def\@mathcolor#1#2#3{%
\protect\leavevmode
\begingroup
\color#1{#2}#3%
\endgroup
}
\makeatother
\def\triadone{brown}
\def\triadtwo{red}
\def\triadthree{blue}
\def\triadfour{green}
\def\triadfive{purple}
\def\triadsix{dimgray}
\def\triadseven{chocolate}
\newtheorem{definition}{Definition}[section]

\newtheorem{lemma}{Lemma}[section]

\newtheorem{remark}{Remark}[section]
\newtheorem{strategy}{Strategy}[section]
\newtheorem{example}{Example}
\usepackage{thmtools}
\usepackage{thm-restate}
\declaretheorem[name=Theorem,numberwithin=section]{thm}
\def\matrixsize#1#2{{{#1}\times{#2}}}
\def\IdentityMatrix{\mathrm{I}}
\def\IdentityMatrixOfSize#1{{\IdentityMatrix}_{#1}}
\def\tensorproduct{\otimes}
\def\tensor#1{{\mathcal{#1}}}

\def\Trace{\textup{Trace}}
\def\InvTranspose#1{{#1}^{-\intercal}}
\def\Inverse#1{{#1}^{-1}}
\def\Isotropy#1{\mathsf{#1}}
\def\IsotropyAction#1#2{{{#1}\diamond{#2}}}
\def\MatrixProduct#1#2{{{#1}\cdot{#2}}}
\def\card#1{|{#1}|}
\def\IsotropyComposition{\circ}

\newenvironment{smatrix}{\left(\begin{smallmatrix}}{\end{smallmatrix}\right)}
\newcommand{\Transpose}[1]{{{#1}^{\intercal}}\xspace}
\newcommand{\CTranspose}[1]{{{\overline{#1}}^{\intercal}}\xspace}
\newcommand{\syrk}{{\textsc{syrk}}\xspace}
\newcommand{\blas}{{\textsc{blas}{\footnotesize{3}}}\xspace}
\newcommand{\mathsc}[1]{{\normalfont\textsc{#1}}}
\newcommand{\Primes}{\ensuremath{\mathbb{P}}}
\newcommand{\Z}{\ensuremath{\mathbb{Z}}}
\newcommand{\K}{\ensuremath{\mathbb{K}}}
\newcommand{\F}{\ensuremath{\mathbb{F}}}
\newcommand{\RR}{\ensuremath{\mathbb{R}}}
\newcommand{\QQ}{\ensuremath{\mathbb{Q}}}
\newcommand{\CC}{\ensuremath{\mathbb{C}}}

\newcommand{\FFSoS}[2]{\ensuremath{\text{\texttt{SoS}}(#1,#2)}}
\newcommand{\FFSqrtname}{\ensuremath{\text{\texttt{sqrt}}}}
\newcommand{\FFSqrt}[2]{\ensuremath{\FFSqrtname(#2)}}

\newcommand{\GO}[1]{\ensuremath{{O\mathopen{}\left({#1}\right)\mathclose{}}}\xspace}
\newcommand{\LO}[1]{\ensuremath{{o\mathopen{}\left({#1}\right)\mathclose{}}}\xspace}
\newcommand{\Contraction}[2]{{\mathopen{}\left\{{#1},{#2}\right\}\mathclose{}}\xspace}
\newcommand{\assign}{\ensuremath{\leftarrow\xspace}}
\makeatletter
\newcommand{\CWHILE}[2][default]{\algorithmicwhile\ #2\ %
\algorithmicdo\hfill%
\ALC@com{#1}\begin{ALC@whl}}
\newcommand{\CFORALL}[2][default]{\ALC@it\algorithmicforall\ #2\ %
\algorithmicdo\newline%
\ALC@com{#1}\begin{ALC@for}}
\newcommand{\CIF}[2][default]{\ALC@it\algorithmicif\ #2\ %
\algorithmicthen\hfill%
\ALC@com{#1}\begin{ALC@if}}
\newcommand{\FORALLDOEND}[2]{\ALC@it\algorithmicforall\ #1\ \algorithmicdo\ #2 \algorithmicendfor}
\makeatother
\usepackage[colorinlistoftodos]{todonotes}
\title{On Fast Multiplication of a Matrix by its Transpose}
\author{Jean-Guillaume Dumas}
\affiliation{%
  \institution{Universit\'e Grenoble Alpes}
  \department{Laboratoire Jean Kuntzmann, CNRS}
  \streetaddress{700 avenue centrale, IMAG --- CS 40700}
  \city{UMR 5224, 38058 Grenoble}
  \country{France}
}
\author{Cl\'ement Pernet}
\affiliation{
  \institution{Universit\'e Grenoble Alpes}
  \department{Laboratoire Jean Kuntzmann, CNRS}
  \streetaddress{700 avenue centrale, IMAG --- CS 40700}
  \city{UMR 5224, 38058 Grenoble}
  \country{France}
}
\author{Alexandre Sedoglavic}
\affiliation{
  \institution{Universit\'e de Lille}
  \department{UMR CNRS 9189 CRISTAL}
  \streetaddress{Cit\'e scientifique}
  \city{59650 Villeneuve d'Ascq}
  \country{France}
}
\makeatletter
\hypersetup{%
	pdftitle={\@title},
	pdfauthor={\@author},
	breaklinks=true,
	colorlinks=true,
	linkcolor=darkred,
	citecolor=blue,
	urlcolor=darkgreen,
}
\makeatother
\usepackage[num]{hypdestopt}

\begin{document}
\begin{abstract}
We present a non-commutative algorithm for the multiplication of
a~${{2}\times{2}}$-block-matrix by its transpose using~$5$ block
products ($3$ recursive calls and~$2$ general products)
over~$\mathbb{C}$ or any field of prime characteristic.
We use geometric considerations on the space of bilinear forms describing~${{2}\times{2}}$ matrix products to obtain this algorithm and we show how to reduce the number of involved additions.
The resulting algorithm for arbitrary dimensions is a reduction of multiplication of a matrix by its transpose to general matrix product, improving by a constant factor previously known reductions.
Finally we propose schedules with low memory footprint that support a fast and memory efficient practical implementation over a prime field.
To conclude, we show how to use our result in~${\MatrixProduct{L}{\MatrixProduct{D}{\Transpose{L}}}}$ factorization.
\end{abstract}

\maketitle
\section{Introduction}\label{sec:prelim}
Strassen's algorithm~\cite{Strassen:1969:GENO}, with~$7$ recursive
multiplications and~$18$ additions, was the first sub-cubic time
algorithm for matrix product, with a cost of~$\GO{n^{2.81}}$.
Summarizing the many improvements which have happened since then, the cost of multiplying two arbitrary~$\matrixsize{n}{n}$ matrices~$\GO{n^{\omega}}$ will be denoted by~${\mathrm{MM}_{\omega}(n)}$ (see~\cite{LeGall:2014:fmm} for the best theoretical value of~$\omega$ known to date).
\par
We propose a new algorithm for the computation of the product~${\MatrixProduct{A}{\Transpose{A}}}$ of a~${\matrixsize{2}{2}}$-block-matrix by its transpose using only~$5$ block multiplications over some base field, instead of~$6$ for the natural divide \& conquer algorithm.
For this product, the best previously known complexity bound was dominated by~${\frac{2}{2^{\omega}-4}\mathrm{MM}_{\omega}(n)}$ over any field (see~\cite[\S~6.3.1]{jgd:2008:toms}).
Here, we establish the following result:
\begin{theorem}\label{thm:main}
The product of an~$\matrixsize{n}{n}$ matrix by its transpose can be computed in~${\frac{2}{2^{\omega}-3}\textrm{MM}_{\omega}(n)}$ field operations over a base field for which there exists a skew-orthogonal matrix.
\end{theorem}
Our algorithm is derived from the class of Strassen-like algorithms multiplying~$\matrixsize{2}{2}$ matrices in~$7$ multiplications.
Yet it is a reduction of multiplying a matrix by its transpose to general matrix multiplication, thus supporting any admissible value for~$\omega$.
By exploiting the symmetry of the problem, it requires about half of the arithmetic cost of general matrix multiplication when~$\omega$ is~$\log_{2}{7}$.
\par
We focus on the computation of the product of an~$\matrixsize{n}{k}$ matrix by its transpose and possibly accumulating the result to another matrix.
Following the terminology of the {\blas} standard~\cite{DDHD90}, this operation is a symmetric rank~$k$ update ({\syrk} for short).
\section{Matrix product algorithms encoded by tensors}
Considered as~$\matrixsize{2}{2}$ matrices, the matrix product~${C=\MatrixProduct{A}{B}}$ could be computed using Strassen algorithm by performing the following computations (see~\cite{Strassen:1969:GENO}):
\begin{equation}
\label{eq:StrassenMultiplicationAlgorithm}
\begin{array}{ll}
\mathcolor{\triadone}{\rho_{1}}\leftarrow{\mathcolor{\triadone}{a_{11}}(\mathcolor{\triadone}{b_{12}-b_{22}})},
&
\\
\mathcolor{\triadtwo}{\rho_{2}}\leftarrow{(\mathcolor{\triadtwo}{a_{11}+a_{12}})\mathcolor{\triadtwo}{b_{22}}},
&
\mathcolor{\triadfour}{\rho_{4}}\leftarrow{(\mathcolor{\triadfour}{a_{12}-a_{22}})(\mathcolor{\triadfour}{b_{21}+b_{22}})},
\\
\mathcolor{\triadthree}{\rho_{3}}\leftarrow{(\mathcolor{\triadthree}{a_{21}+a_{22}}) \mathcolor{\triadthree}{b_{11}}},
&
\mathcolor{\triadfive}{\rho_{5}}\leftarrow{(\mathcolor{\triadfive}{a_{11}+a_{22}})(\mathcolor{\triadfive}{b_{11}+b_{22}})},
\\
\mathcolor{\triadsix}{\rho_{6}}\leftarrow{\mathcolor{\triadsix}{a_{22}}(\mathcolor{\triadsix}{b_{21}-b_{11}})},
&
\mathcolor{\triadseven}{\rho_{7}}\leftarrow{(\mathcolor{\triadseven}{a_{21}-a_{11}})(\mathcolor{\triadseven}{b_{11}+b_{12}})},
\\[\smallskipamount]
\multicolumn{2}{c}{
\begin{smatrix} c_{11} &c_{12} \\ c_{21} &c_{22} \end{smatrix}
=
\begin{smatrix}
\mathcolor{\triadfive}{\rho_{5}} + \mathcolor{\triadfour}{\rho_{4}} - \mathcolor{\triadtwo}{\rho_{2}} + \mathcolor{\triadsix}{\rho_{6}} &
\mathcolor{\triadsix}{\rho_{6}} + \mathcolor{\triadthree}{\rho_{3}} \\
\mathcolor{\triadtwo}{\rho_{2}} + \mathcolor{\triadone}{\rho_{1}}&
\mathcolor{\triadfive}{\rho_{5}} + \mathcolor{\triadseven}{\rho_{7}} + \mathcolor{\triadone}{\rho_{1}}- \mathcolor{\triadthree}{\rho_{3}}
\end{smatrix}\!.}
\end{array}
\end{equation}
In order to consider this algorithm under a geometric standpoint, we present it as a tensor.
Matrix multiplication is a bilinear map:
\begin{equation}
\begin{array}{ccl}
\K^{\matrixsize{m}{n}} \times \K^{\matrixsize{n}{p}} & \rightarrow &\K^{\matrixsize{m}{p}}, \\
(X,Y) &\rightarrow & \MatrixProduct{X}{Y},
\end{array}
\end{equation}
where the spaces~$\K^{\matrixsize{a}{b}}$ are finite vector spaces that can be endowed with the Frobenius inner product~${{\langle M,N\rangle}={\Trace({\MatrixProduct{\Transpose{M}}{N}})}}$.
Hence, this inner product establishes an isomorphism between~$\K^{\matrixsize{a}{b}}$ and its dual space~$\bigl(\K^{\matrixsize{a}{b}}\bigr)^{\star}$ allowing for example to associate matrix multiplication and the trilinear form~${\Trace(\MatrixProduct{\Transpose{Z}}{\MatrixProduct{X}{Y}})}$:
\begin{equation}
\label{eq:TrilinearForm}
\begin{array}{ccc}
\K^{\matrixsize{m}{n}} \times \K^{\matrixsize{n}{p}} \times {(\K^{\matrixsize{m}{p}})}^{\star}&\rightarrow & \K, \\
(X,Y,\Transpose{Z}) &\rightarrow & \langle {Z},\MatrixProduct{X}{Y}\rangle.
\end{array}
\end{equation}
As by construction, the space of trilinear forms is the canonical dual space of order three tensor product, we could associate the Strassen multiplication algorithm~(\ref{eq:StrassenMultiplicationAlgorithm}) with the tensor~$\tensor{S}$ defined by:
\begin{equation}
\label{eq:StrassenTensor}
\begin{array}{r}
\sum_{i=1}^{7}{S_{i1}}\!\tensorproduct\!{S_{i2}}\!\tensorproduct\!{S_{i3}}=
\mathcolor{\triadone}{\begin{smatrix}1&0\\0&0\\\end{smatrix}\!\tensorproduct\!\begin{smatrix}0&1\\0&-1\\\end{smatrix}\!\tensorproduct\!\begin{smatrix}0&0\\1&1\\\end{smatrix}
}
\!+\!\\[\bigskipamount]
\mathcolor{\triadtwo}{\begin{smatrix}1&1\\0&0\\\end{smatrix}\!\tensorproduct\!\begin{smatrix}0&0\\0&1\\\end{smatrix}\!\tensorproduct\!\begin{smatrix}-1&0\\1&0\\\end{smatrix}}
\!+\!
\mathcolor{\triadthree}{\begin{smatrix}0&0\\1&1\\\end{smatrix}\!\tensorproduct\!\begin{smatrix}1&0\\0&0\\\end{smatrix}\!\tensorproduct\!\begin{smatrix}0&1\\0&-1\end{smatrix}}
\!+\!\\[\bigskipamount]
\mathcolor{\triadfour}{\begin{smatrix}0&1\\0&-1\\\end{smatrix}\!\tensorproduct\!\begin{smatrix}0&0\\1&1\\\end{smatrix}\!\tensorproduct\!\begin{smatrix}1&0\\0&0\\\end{smatrix}}
\!+\!
\mathcolor{\triadfive}{{\begin{smatrix}1&0\\0&1\end{smatrix}}\!\tensorproduct\!{\begin{smatrix}1&0\\0&1\end{smatrix}}\!\tensorproduct\!\begin{smatrix}1&0\\0&1\\\end{smatrix}}
\!+\!\\[\bigskipamount]
\mathcolor{\triadsix}{\begin{smatrix}0&0\\0&1\\\end{smatrix}\!\tensorproduct\!\begin{smatrix}-1&0\\1&0\\\end{smatrix}\!\tensorproduct\!\begin{smatrix}1&1\\0&0\\\end{smatrix}}
\!+\!
\mathcolor{\triadseven}{\begin{smatrix}-1&0\\1&0\\\end{smatrix}\!\tensorproduct\!\begin{smatrix}1&1\\0&0\\\end{smatrix}\!\tensorproduct\!\begin{smatrix}0&0\\0&1\\\end{smatrix}}
\!
\end{array}
\end{equation}
in~${{(\K^{\matrixsize{m}{n}})}^{\star} \tensorproduct {(\K^{\matrixsize{n}{p}})}^{\star} \tensorproduct \K^{\matrixsize{m}{p}}}$ with~${m=n=p=2}$.
Given any couple~$(A,B)$ of~$\matrixsize{2}{2}$-matrices, one can explicitly retrieve from tensor~$\tensor{S}$ the Strassen matrix multiplication algorithm computing~$\MatrixProduct{A}{B}$ by the \emph{partial} contraction~${\Contraction{\tensor{S}}{A\tensorproduct B}}$:
\begin{equation}
\label{eq:TensorAction}
\begin{array}{c}
\left({(\K^{\matrixsize{m}{n}})}^{\star}\! \tensorproduct \! {(\K^{\matrixsize{n}{p}})}^{\star}\! \tensorproduct\! \K^{\matrixsize{m}{p}} \right)\! \tensorproduct \! \left(
\K^{\matrixsize{m}{n}}\! \tensorproduct\! \K^{\matrixsize{n}{p}} \right)
\!\rightarrow\!
\K^{\matrixsize{m}{p}}, \\[\smallskipamount]
\tensor{S}\tensorproduct (A \tensorproduct B) \rightarrow
\sum_{i=1}^{7} \langle {S_{i1}}, A \rangle
\langle {S_{i2}}, B \rangle S_{i3},
\end{array}
\end{equation}
while the \emph{complete} contraction~${\Contraction{\tensor{S}}{{A}\tensorproduct{B}\tensorproduct{\Transpose{C}}}}$ is~${\Trace(\MatrixProduct{\MatrixProduct{A}{B}}{C})}$.
\par
The tensor formulation of matrix multiplication algorithm gives explicitly its symmetries (a.k.a.\ \emph{isotropies}).
As this formulation is associated to the trilinear form~$\Trace(\MatrixProduct{\MatrixProduct{A}{B}}{C})$, given three invertible matrices~$U,V,W$ of suitable sizes and the classical properties of the trace, one can remark that~${\Trace(\MatrixProduct{A}{\MatrixProduct{B}{C}})}$ is equal to:
	\begin{equation}
		\begin{array}{l}
	\label{eq:isotropy}
  \Trace\bigl(\Transpose{(\MatrixProduct{\MatrixProduct{A}{B}}{C})}\bigr)
  =\Trace(\MatrixProduct{C}{\MatrixProduct{A}{B}})
	=\Trace(\MatrixProduct{B}{\MatrixProduct{C}{A}}),\\
	\textrm{and}\ \Trace\bigl(\MatrixProduct{\Inverse{U}}{\MatrixProduct{A}{V}}
  \cdot \Inverse{V} \cdot B \cdot W \cdot \Inverse{W} \cdot C \cdot U\bigr).
		\end{array}
	\end{equation}
These relations illustrate the following theorem:
\begin{theorem}[{\cite[\S~2.8]{groot:1978a}}]
The isotropy group of the~$\matrixsize{n}{n}$ matrix multiplication tensor is~${{{\mathsc{psl}^{\pm}({\K^{n}})}^{\times 3}}\!\rtimes{\mathfrak{S}_{3}}}$, where~$\mathsc{psl}$ stands for the group of matrices of determinant~${\pm{1}}$ and~$\mathfrak{S}_{3}$ for the symmetric group on~$3$ elements.
\end{theorem}
The following definition recalls the \emph{sandwiching} isotropy on matrix multiplication tensor:
\begin{definition}\label{def:sandwiching}
Given~${\Isotropy{g}={(U\times V \times W)}}$ in~${\mathsc{psl}^{\pm}({\K}^{n})}^{\times 3}$, its action~${\IsotropyAction{\Isotropy{g}}{\tensor{S}}}$ on a tensor~$\tensor{S}$ is given by~${\sum_{i=1}^{7} \IsotropyAction{\Isotropy{g}}{(S_{i1}\tensorproduct{} S_{i2}\tensorproduct{} S_{i3})}}$ where the term~${\IsotropyAction{\Isotropy{g}}{(S_{i1}\tensorproduct{} S_{i2}\tensorproduct{} S_{i3})}}$ is equal to:
\begin{equation}
\label{eq:sandwiching}
{\left(\MatrixProduct{\InvTranspose{U}}{\MatrixProduct{S_{i1}}{\Transpose{V}}}\right)}
\tensorproduct
{\left(\MatrixProduct{\InvTranspose{V}}{\MatrixProduct{S_{i2}}{\Transpose{W}}}\right)}
\tensorproduct
{\left(\MatrixProduct{\InvTranspose{W}}{\MatrixProduct{S_{i3}}{\Transpose{U}}}\right)}.
\end{equation}
\end{definition}
\begin{remark}\label{rem:groupcomposition}
In~${\mathsc{psl}^{\pm}({\K^{n}})}^{\times 3}$, the product~$\IsotropyComposition$ of two isotropies~$g_{1}$ defined by~${{u_{1}}\times{v_{1}}\times{w_{1}}}$ and~${g_{2}}$ by~${{u_{2}}\times{v_{2}}\times{w_{2}}}$ is the isotropy~${g_{1}\IsotropyComposition g_{2}}$ equal to~${\MatrixProduct{u_{1}}{u_{2}}\times{}\MatrixProduct{v_{1}}{v_{2}}\times{}\MatrixProduct{w_{1}}{w_{2}}}$.
Furthermore,the complete contraction~${\Contraction{{g_{1}}\IsotropyComposition{g_{2}}}{{A}\tensorproduct{B}\tensorproduct{C}}}$ is equal to~${\Contraction{g_{2}}{\IsotropyAction{\Transpose{g_{1}}}{{A}\tensorproduct{B}\tensorproduct{C}}}}$.
\end{remark}
The following theorem shows that all~$\matrixsize{2}{2}$-matrix product algorithms with~$7$ coefficient multiplications could be obtained by the action of an isotropy on Strassen tensor:
\begin{theorem}[{\cite[\S~0.1]{groot:1978}}]
The group~${{\mathsc{psl}^{\pm}({\K^{n}})}^{\times 3}}$ acts transitively on the variety of optimal algorithms for the computation of~$\matrixsize{2}{2}$-matrix multiplication.
\end{theorem}
Thus, isotropy action on Strassen tensor may define other matrix
product algorithm with interesting computational properties.
\subsection{Design of a specific~\texorpdfstring{$\matrixsize{2}{2}$}{2x2}-matrix product}
This observation inspires our general strategy to design specific algorithms suited for particular matrix product.
\begin{strategy}\label{strategy}
By applying an undetermined isotropy:
\begin{equation}
\label{eq:GenericIsotropy}
\Isotropy{g} = U \times V \times W =
\begin{smatrix} u_{11} & u_{12} \\ u_{21} & u_{22}\end{smatrix} \times
\begin{smatrix} v_{11} & v_{12} \\ v_{21} & v_{22}\end{smatrix} \times
\begin{smatrix} w_{11} & w_{12} \\ w_{21} & w_{22}\end{smatrix}
\end{equation}
on Strassen tensor~$\tensor{S}$, we obtain a parameterization~${\tensor{T}=\IsotropyAction{\Isotropy{g}}{\tensor{S}}}$ of all matrix product algorithms requiring~$7$ coefficient multiplications:
\begin{equation}
\label{eq:optimal2x2algoparam}
\tensor{T}=\sum_{i=1}^{7} T_{i1}\tensorproduct T_{i2}\tensorproduct T_{i3},\
T_{i1}\tensorproduct T_{i2}\tensorproduct T_{i3}
= \IsotropyAction{\Isotropy{g}}{S_{i1}\tensorproduct S_{i2}\tensorproduct S_{i3}}.
\end{equation}
Then, we could impose further conditions on these algorithms and check by a Gr\"{o}bner basis computation if such an algorithm exists.
If so, there is subsequent work to do for choosing a point on this variety;
this choice can be motivated by the additive cost bound and the scheduling property of the evaluation scheme given by this point.
\end{strategy}
Let us first illustrate this strategy with the well-known Winograd variant of Strassen algorithm presented in~\cite{Winograd:1977:complexite}.
\begin{example}\label{ex:bshouty}
Apart from the number of multiplications, it is also interesting in practice to reduce the number of additions in an algorithm.
Matrices~$S_{11}$ and~$S_{61}$ in tensor~(\ref{eq:StrassenTensor}) do not increase the additive cost bound of this algorithm.
Hence, in order to reduce this complexity in an algorithm, we could try to maximize the number of such matrices involved in the associated tensor.
To do so, we recall Bshouty's results on additive complexity of matrix product algorithms.
\begin{theorem}[\cite{bshouty:1995a}]
Let~${e_{(i,j)}=(\delta_{i,k}\delta_{j,l})_{(k,l)}}$ be the single entry elementary matrix.
A~$\matrixsize{2}{2}$ matrix product tensor could not have~$4$ such matrices as first (resp.\ second, third) component~(\cite[Lemma~8]{bshouty:1995a}).
The additive complexity bound of first and second components are equal~(\cite[eq.~(11)]{bshouty:1995a}) and at least~${4=7-3}$.
The total additive complexity of~$\matrixsize{2}{2}$-matrix product is at least~$15$~(\cite[Theorem~1]{bshouty:1995a}).
\end{theorem}
Following our strategy, we impose on tensor~$\tensor{T}$~(\ref{eq:optimal2x2algoparam}) the constraints
\begin{equation}
\label{eq:WinogradConstraints}
{T_{11}=e_{1,1}=\begin{smatrix} 1&0\\0&0\end{smatrix}\!,\quad T_{12}=e_{1,2},\quad T_{13}=e_{2,2}}
\end{equation}
and obtain by a Gr\"obner basis computation~\cite{FGb} that such tensors are the images of Strassen tensor by the action of the following isotropies:
\begin{equation}
\label{eq:IsotropyFromStrassen2Wingrad}
\Isotropy{w}=
	\begin{smatrix}1&0\\0&1\end{smatrix} \times
\begin{smatrix}1&-1\\0&-1\end{smatrix} \times
	\begin{smatrix}w_{11}&w_{12}\\w_{21}&w_{22}\end{smatrix}\!.
\end{equation}
The variant of the Winograd tensor~\cite{Winograd:1977:complexite}
presented with a renumbering as \cref{alg:sw} is obtained by the action of~$\Isotropy{w}$ with the specialization~${w_{12}=w_{21}=1=-w_{11},w_{22}=0}$ on the Strassen tensor~$\tensor{S}$.
While the original Strassen algorithm requires~$18$ additions, only~$15$ additions are necessary in the Winograd \cref{alg:sw}.
\end{example}
\begin{algorithm}[htbp]\caption{: $C=\mathrm{W}(A,B)$}\label{alg:sw}
\begin{algorithmic}
\Require{${A=\begin{smatrix} a_{11}&a_{12}\\ a_{21}&a_{22}\end{smatrix}}$ and~${B=\begin{smatrix} b_{11}&b_{12}\\ b_{21}&b_{22}\end{smatrix}}$;}
\Ensure{${C=\MatrixProduct{A}{B}}$}
\State{%
  \( \renewcommand{\arraycolsep}{.7mm}
  \begin{array}{LLLL}
   s_{1} \leftarrow \mathcolor{\triadfour}{a_{11} - a_{21}}, &{}
   s_{2} \leftarrow \mathcolor{\triadseven}{a_{21} + a_{22}}, &{}
   s_{3} \leftarrow \mathcolor{\triadfive}{s_{2} - a_{11}}, &{}
   s_{4} \leftarrow \mathcolor{\triadsix}{a_{12} - s_{3}}, \\
   t_{1} \leftarrow \mathcolor{\triadfour}{b_{22} - b_{12}}, &{}
   t_{2} \leftarrow \mathcolor{\triadseven}{b_{12} - b_{11}}, &{}
   t_{3} \leftarrow \mathcolor{\triadfive}{b_{11} + t_{1}}, &{}
   t_{4} \leftarrow \mathcolor{\triadthree}{b_{21} - t_{3}}.\\
 \end{array}
\)
}
\State{%
  \(
  \renewcommand{\arraycolsep}{.7mm}
  \begin{array}{LLLL}
\mathcolor{\triadone}{p_{1}\leftarrow{a_{11}}{\cdot}{b_{11}}},&
\mathcolor{\triadtwo}{p_{2}\leftarrow{a_{12}}{\cdot}{b_{21}}},&
\mathcolor{\triadthree}{p_{3}\leftarrow{a_{22}}{\cdot}{t_{4}}}, &
\mathcolor{\triadfour}{p_{4}\leftarrow{s_{1}}{\cdot}{t_{1}}},\\
&
\mathcolor{\triadfive}{p_{5}\leftarrow{s_{3}}{\cdot}{t_{3}}}, &
\mathcolor{\triadsix}{p_{6}\leftarrow{s_{4}}{\cdot}{b_{22}}}, &
\mathcolor{\triadseven}{p_{7}\leftarrow{s_{2}}{\cdot}{t_{2}}}.\\
\end{array}
\)
}
\State{%
 \(    \renewcommand{\arraycolsep}{.7mm}
 \begin{array}{LLLL}
    c_{1} \leftarrow \mathcolor{\triadone}{p_{1}} + \mathcolor{\triadfive}{p_{5}}, &
    c_{2} \leftarrow c_{1} + \mathcolor{\triadfour}{p_{4}}, &
    c_{3} \leftarrow \mathcolor{\triadone}{p_{1}} + \mathcolor{\triadtwo}{p_{2}}, &
    c_{4} \leftarrow c_{2} + \mathcolor{\triadthree}{p_{3}}, \\
    c_{5} \leftarrow c_{2} + \mathcolor{\triadseven}{p_{7}}, &
    c_{6} \leftarrow c_{1} + \mathcolor{\triadseven}{p_{7}}, &
    c_{7} \leftarrow c_{6} + \mathcolor{\triadsix}{p_{6}}.&
  \end{array}
 \)
}
\Return {$ C = \begin{smatrix} c_{3} & c_{7} \\ c_{4} & c_{5} \end{smatrix}$.}
\end{algorithmic}
\end{algorithm}
As a second example illustrating our strategy, we consider now the matrix squaring that was already explored by Bodrato in~\cite{Bodrato:2010:square}.
\begin{example}
When computing~$A^{2}$, the contraction~(\ref{eq:TensorAction}) of the
tensor~$\tensor{T}$~(\ref{eq:optimal2x2algoparam})
with~${{A}\tensorproduct{A}}$ shows that choosing a subset~$J$
of~${\{1,\ldots,7\}}$ and imposing~${T_{i1}=T_{i2}}$ as constraints
with~$i$ in~$J$ (see~\cite[eq~4]{Bodrato:2010:square}) can
save~$\card{J}$ operations and thus reduce the computational
complexity.
\par
The definition~(\ref{eq:optimal2x2algoparam}) of\,~$\tensor{T}$, these constraints, and the fact that $U, V$ and~$W$'s determinant are~$1$, form a system with~${3+4\,\card{J}}$ equations and~$12$ unknowns whose solutions define matrix squaring algorithms.
\par
The algorithm~\cite[\S~2.2, eq~2]{Bodrato:2010:square} is given by the action of the isotropy:
\begin{equation}
	\Isotropy{g} =
	\begin{smatrix} 0 & 1 \\ -1 & 0 \end{smatrix} \times
	\begin{smatrix} 1 & 1 \\ 0 & 1 \end{smatrix} \times
	\begin{smatrix} 1 & 0 \\ 1 & 1 \end{smatrix}
\end{equation}
on Strassen's tensor and is just Chatelin's algorithm~\cite[Appendix~A]{Chatelin:1986:transformations}, with~${\lambda=1}$ (published~$25$ years before~\cite{Bodrato:2010:square}, but not applied to squaring).
\end{example}
\begin{remark}
Using symmetries in our strategy reduces the computational cost compared to the resolution of Brent's equations~\cite[\S~5, eq~5.03]{brent:1970a} with an undetermined tensor~$\tensor{T}$.
In the previous example by doing so, we should have constructed a system of at most~$64$ algebraic equations with~${{4(3\, (7-\card{J})+2\,\card{J})}}$ unknowns, resulting from the constraints on~$\tensor{T}$ and the relation~${\tensor{T}=\tensor{S}}$, expressed using Kronecker product as a single zero matrix in~$\K^{\matrixsize{8}{8}}$.
\end{remark}
We apply now our strategy on the~$\matrixsize{2}{2}$ matrix product~$\MatrixProduct{A}{\Transpose{A}}$.
\subsection{\texorpdfstring{${{2}\times{2}}$}{2x2}-matrix product by its transpose}\label{sec:GramMultAlg}
Applying our Strategy~\ref{strategy}, we consider~(\ref{eq:optimal2x2algoparam}) a generic matrix multiplication tensor~${\tensor{T}}$ and our goal is to reduce the computational complexity of the partial contraction~(\ref{eq:TensorAction}) with~${{A}\tensorproduct{\Transpose{A}}}$ computing~$\MatrixProduct{A}{\Transpose{A}}$.
\par
By the properties of the transpose operator and the trace, the following relations hold:
\begin{equation}
\bigl\langle {T_{i2}}, \Transpose{A} \bigr\rangle
\begin{array}[t]{l}
	=\Trace\bigl(\MatrixProduct{\Transpose{T_{i2}}}{\Transpose{A}}\bigr)
	=\Trace\bigl(\Transpose{(\MatrixProduct{A}{T_{i2}})}\bigr),\\
	=\Trace\bigl(\MatrixProduct{A}{T_{i2}}\bigr)
	=\Trace\bigl(\MatrixProduct{T_{i2}}{A}\bigr)
        = \bigl\langle {\Transpose{T_{i2}}}, A \bigr\rangle.
\end{array}
\end{equation}
Thus, the partial contraction~(\ref{eq:TensorAction}) satisfies here the following relation:
\begin{equation}
\label{eq:MainExpression}
\sum_{i=1}^{7} \bigl\langle T_{i1}, A \bigr\rangle \bigl\langle T_{i2}, \Transpose{A} \bigr\rangle T_{i3} =\sum_{i=1}^{7} \bigl\langle T_{i1}, A \bigr\rangle \langle \Transpose{T_{i2}}, A \rangle T_{i3}.
\end{equation}
\subsubsection{Supplementary symmetry constraints}
Our goal is to save computations in the evaluation of~(\ref{eq:MainExpression}).
To do so, we consider the subsets~${J}$ of~${\{1,\ldots,7\}}$ and~${H}$ of~${\left\{{(i,j)}\in{{\{2,\ldots,7\}}^{2}} | i\not =j, i\not\in J, j\not\in J\right\}}$ in order to express the following constraints:
\begin{equation}
	\label{eq:aatconstraints}
	T_{i1}=\Transpose{T_{i2}},\ i\in J,\quad T_{j1}=\Transpose{T_{k2}},\ T_{k1}=\Transpose{T_{j2}}, \ (j,k)\in H.
\end{equation}
The constraints of type~$J$ allow one to save preliminary additions
when applying the method to matrices~${B=\Transpose{A}}$: since then operations on~$A$ and~$\Transpose{A}$ will be the same.
The constraints of type~$H$ allow to save multiplications especially
when dealing with a block-matrix product: in fact, if some matrix
products are transpose of another, only one of the pair needs to
be computed as shown in Section~\ref{sec:GamBlocMatrixProduct}.
\par
We are thus looking for the largest possible sets~$J$ and~$H$.
By exhaustive search, we conclude that the cardinality of~$H$ is at
most~$2$ and then the cardinality of~$J$ is at most~$3$.
For example, choosing the sets~${J=\{1,2,5\}}$ and~${H=\{(3,6),(4,7)\}}$ we obtain for these solutions the following parameterization expressed with a primitive element~${z=v_{11}-v_{21}}$:
\begin{equation}
\label{eq:SolutionsParameterization}
\begin{array}{ccl}
	v_{11}&=&z+v_{21}, \\
	v_{22}&=& \bigl({2\,v_{21}}{({v_{21}}+z)} -1\bigr){v_{21}} +z^{3},\\
	v_{12}&=& -\bigl({v_{21}}^{2}+{({v_{21}}+z^{2})}^{2}+1\bigr){v_{21}}-z,\\
	u_{11} &=& -\bigl({(z+v_{21})}^2+{v_{21}}^{\!2}\bigr) (w_{21}+w_{22}), \\
	u_{21} &=& -\bigl({(z+v_{21})}^2+{v_{21}}^{\!2}\bigr) (w_{11}+w_{12}), \\
	u_{12} &=& -\bigl({(z+v_{21})}^2+{v_{21}}^{\!2}\bigr) w_{22}, \\
	u_{22} &=& \bigl({(z+v_{21})}^2+{v_{21}}^{\!2}\bigr) w_{12},
	\\[\smallskipamount]
	\multicolumn{3}{c}{{\bigl({(z+v_{21})}^{2}+{v_{21}}^{\!2}\bigr)}^{2}+1=0,\ {w_{11}w_{22}-w_{12}w_{21}=1.}}
\end{array}
\end{equation}
\begin{remark}
As~${{\bigl({(z+v_{21})}^{2}+{v_{21}}^{\!2}\bigr)}^{2}+1=0}$ occurs in this parameterization, field extension could not be avoided in these algorithms if the field does not have---at least---a square root of~${-1}$.
We show in Section~\ref{sec:GamBlocMatrixProduct} that we can avoid
these extensions with block-matrix products and use our algorithm
directly in any field of prime characteristic.
\end{remark}
\subsubsection{Supplementary constraint on the number of additions}
As done in Example~\ref{ex:bshouty}, we could also try to reduce the additive complexity and use~$4$ pre-additions on~$A$ (resp.~$B$)~\cite[Lemma~9]{bshouty:1995a} and~$7$ post-additions on the products to form~$C$~\cite[Lemma~2]{bshouty:1995a}.
In the current situation, if the operations on~$B$ are exactly the transpose of that of~$A$, then we have the following lower bound:
\begin{lemma}\label{lem:eleven}
Over a non-commutative domain,~$11$ additive operations are necessary to multiply a~$\matrixsize{2}{2}$ matrix by its transpose with a bilinear algorithm that uses~$7$ multiplications.
\end{lemma}
Indeed, over a commutative domain, the lower left and upper right parts of the product are transpose of one another and one can save also multiplications.
Differently, over non-commutative domains,~$\MatrixProduct{A}{\Transpose{A}}$ is not symmetric in general (say~${{ac+bd}\neq{ca+db}}$) and all four coefficients need to be computed.
But one can still save~$4$ additions, since there are algorithms where
pre-additions are the same on~$A$ and~$\Transpose{A}$.
Now, to reach that minimum, the constraints~(\ref{eq:aatconstraints}) must be combined with the minimal number~$4$ of pre-additions for~$A$.
Those can be attained only if~$3$ of the~$T_{i1}$ factors do not require any addition~\cite[Lemma~8]{bshouty:1995a}.
Hence, those factors involve only one of the four elements of~$A$ and they are just permutations of~$e_{11}$.
We thus add these constraints to the system for a subset~$K$
of~${\{1,\ldots,7\}}$:
\begin{equation}
	\label{eq:addconstraints}
	\card{K}=3~\text{and}~T_{i1}\ \textrm{is in}\ \left\{
	\begin{smatrix} 1&0\\0&0\end{smatrix},
	\begin{smatrix} 0&1\\0&0\end{smatrix},
	\begin{smatrix} 0&0\\1&0\end{smatrix},
	\begin{smatrix} 0&0\\0&1\end{smatrix}
		\right\}\ \textrm{and}\ i\ \textrm{in}\ K.
\end{equation}
\subsubsection{Selected solution}
We choose~${K=\{1,2,3\}}$ similar to~(\ref{eq:WinogradConstraints}) and obtain the following isotropy that sends Strassen tensor to an algorithm computing the symmetric product more efficiently:
\begin{equation}
\label{eq:ChoosenAAT}
	\Isotropy{a}=
	\begin{smatrix}z^{2}&0\\0&z^{2}\end{smatrix} \times
	\begin{smatrix}z&-z\\0&z^{3}\end{smatrix} \times
	\begin{smatrix}-1&1\\1&0\end{smatrix},\quad z^{4}=-1.
\end{equation}
We remark that~${\Isotropy{a}}$ is equal to~${\Isotropy{d}\IsotropyComposition\Isotropy{w}}$ with~$\Isotropy{w}$ the isotropy~(\ref{eq:IsotropyFromStrassen2Wingrad}) that sends Strassen tensor to Winograd tensor and with:
\begin{equation}
	\Isotropy{d}={D_{1}}\tensorproduct{D_{2}}\tensorproduct{D_{3}}=
	\begin{smatrix}z^{2}&0\\0&z^{2}\end{smatrix} \times
	\begin{smatrix}z&0\\0&-z^{3}\end{smatrix} \times
	\begin{smatrix}1&0\\0&1\end{smatrix},\ z^{4}=-1.
\end{equation}
Hence, the induced algorithm can benefit from the scheduling and additive complexity of the classical Winograd algorithm.
In fact, our choice~${\IsotropyAction{\Isotropy{a}}{\tensor{S}}}$ is equal to~${\IsotropyAction{(\Isotropy{d}\IsotropyComposition\Isotropy{w})}{\tensor{S}}}$ and thus, according to remark~(\ref{rem:groupcomposition}) the resulting algorithm expressed as the total contraction
\begin{equation}
	\Contraction{\IsotropyAction{(\Isotropy{d}\IsotropyComposition\Isotropy{w})}{\tensor{S}}}{({A}\tensorproduct{\Transpose{A}}\tensorproduct{C})}=
	\Contraction{\IsotropyAction{\Isotropy{w}}{\tensor{S}}}{\IsotropyAction{\Transpose{d}}{({A}\tensorproduct{\Transpose{A}}\tensorproduct{C})}}
\end{equation}
could be written as a slight modification of Algorithm~\ref{alg:sw} inputs.
\par
Precisely, as~$\Isotropy{d}$'s components are diagonal, the relation~${\Transpose{\Isotropy{d}}=\Isotropy{d}}$ holds; hence, we could express input modification as:
\begin{equation}
\label{eq:FromGB2ModifiedWinograd}
{\left(\MatrixProduct{\Inverse{D_{1}}}{\MatrixProduct{A}{D_{2}}} \right)}
\tensorproduct
{\left(\MatrixProduct{\Inverse{D_{2}}}{\MatrixProduct{\Transpose{A}}{D_{3}}}\right)}
\tensorproduct
{\left(\MatrixProduct{\Inverse{D_{3}}}{\MatrixProduct{C}{D_{1}}} \right)}.
\end{equation}
The above expression is trilinear and the matrices~$D_{i}$ are scalings of the identity for~$i$ in~${\{1,3\}}$, hence our modifications are just:
\begin{equation}
\label{eq:InputsFromGB2ModifiedWinograd}
{\left({\frac{1}{z^{2}}\MatrixProduct{A}{D_{2}}} \right)}
\tensorproduct
{\left(\MatrixProduct{\Inverse{D_{2}}}{\Transpose{A}}\right)}
\tensorproduct
{{z^{2}}{C}}.
\end{equation}
Using notations of Algorithm~\ref{alg:sw}, this is~${C=\mathrm{W}\bigl(\MatrixProduct{A}{D_{2}},\MatrixProduct{\Inverse{D_{2}}}{\Transpose{A}}\bigr)}$.
\par
Allowing our isotropies to have determinant different from~$1$, we rescale~$D_{2}$ by a factor~$1/z$ to avoid useless~$4$th root as follows:
\begin{equation}
\label{eq:Y}
Q=\frac{D_{2}}{z} =\begin{smatrix}1&0\\0&-z^{2}\end{smatrix} =\begin{smatrix}1&0\\0&-y\end{smatrix} \!,\quad z^{4}=-1
\end{equation}
where~$y$ designates the expression~${z^{2}}$ that is a root
of~$-1$.
Hence, our algorithm to compute the symmetric product is:
\begin{equation}
\label{eq:2x2GramAlgo}
	C=\mathrm{W}\left(\MatrixProduct{A}{\frac{D_{2}}{z}},\MatrixProduct{\Inverse{\left(\frac{D_{2}}{z}\right)}}{\Transpose{A}}\right)=\mathrm{W}\!\left(\MatrixProduct{A}{Q},\Transpose{\left(\MatrixProduct{A}{\Transpose{(\Inverse{Q})}}\right)}\right)\!.
\end{equation}
In the next sections, we describe and extend this algorithm to higher-dimensional symmetric products~${\MatrixProduct{A}{\Transpose{A}}}$ with a~$\matrixsize{2^{\ell}m}{2^{\ell}m}$ matrix~$A$.
\section{Fast~\texorpdfstring{$\matrixsize{2}{2}$}{2x2}-block recursive \syrk}\label{sec:GamBlocMatrixProduct}
The algorithm presented in the previous section is non-commutative and thus we can extend it to higher-dimensional matrix product by a divide and conquer approach.
To do so, we use in the sequel upper case letters for coefficients in
our algorithms instead of lower case previously (since these
coefficients now represent matrices).
Thus, new properties and results are induced by this shift of perspective.
For example, the coefficient~$Y$ introduced in~(\ref{eq:Y}) could now be transposed in~(\ref{eq:2x2GramAlgo}); that leads to the following definition:
\begin{definition}
An invertible matrix is \emph{skew-orthogonal} if the following relation~${\Transpose{Y}= -\Inverse{Y}}$ holds.
\end{definition}
If~$Y$ is skew-orthogonal, then of the~$7$ recursive matrix products involved in expression~(\ref{eq:2x2GramAlgo}):~$1$ can be avoided ($P_{6}$) since we do not need the upper right coefficient anymore,~$1$ can be avoided since it is the transposition of another product (${P_{7}=\Transpose{P_{4}}}$) and~$3$ are recursive calls to {\syrk}.
This results in Algorithm~\ref{alg:wishartp}.

\begin{algorithm}[htbp]\caption{\texttt{syrk}: symmetric
    matrix product
  }\label{alg:wishartp}
\begin{spacing}{1}
\begin{algorithmic}
\Require{$A=\begin{smatrix} A_{11}&A_{12}\\ A_{21}&A_{22}\end{smatrix}$; a skew-orthogonal matrix~$Y$.}
\Ensure{The lower left triangular part of~$C=\MatrixProduct{A}{\Transpose{A}}=\begin{smatrix}
  C_{11}& \Transpose{C_{21}}\\ {C_{21}} &C_{22}\end{smatrix}$.}

\State \Comment{4 additions and 2 multiplications by~$Y$:}
\State \({S_{1}}\leftarrow{\MatrixProduct{(A_{21} - A_{11})}{Y}}, \hspace{.8em} {S_{2}}\leftarrow{A_{22} - \MatrixProduct{A_{21}}{Y}}\),
\State \({S_{3}}\leftarrow{S_{1} - A_{22}},\hspace{3.4em} {S_{4}}\leftarrow{S_{3} + A_{12}}\).
\State \Comment{3 recursive {\syrk} (${P_{1}, P_{2}, P_{5}}$) and~$2$ generic (${P_{3}, P_{4}}$) products:}
\State \(\mathcolor{\triadone}{P_{1} \leftarrow\MatrixProduct{A_{11}}{\Transpose{A_{11}}}}, \hspace{2.5em}
\mathcolor{\triadtwo}{P_{2} \leftarrow\MatrixProduct{A_{12} }{\Transpose{A_{12}}}}\),
\State \(\mathcolor{\triadthree}{P_{3} \leftarrow\MatrixProduct{A_{22}}{\Transpose{S_{4}}}}, \hspace{3em}
\mathcolor{\triadfour}{P_{4} \leftarrow\MatrixProduct{S_{1}}{\Transpose{S_{2}}}}, \hspace{3em}
\mathcolor{\triadfive}{P_{5} \leftarrow\MatrixProduct{S_{3}}{\Transpose{S_{3}}}}.\)
\State \Comment{2 symmetric additions (half additions):}
\State \(\text{Low}(U_{1})\!\leftarrow\!\text{Low}(P_{1})\!+\!\text{Low}(P_{5})\), \Comment{$U_{1}, P_{1}, P_{5}$~\text{are symm.}}
\State \(\text{Low}(U_{3})\!\leftarrow\!\text{Low}(P_{1})\!+\!\text{Low}(P_{2})\), \Comment{$U_{3}, P_{1}, P_{2}$~\text{are symm.}}

\State \Comment{2 complete additions ($P_{4}$ and~$P_{3}$ are not symmetric):}
\State \(\text{Up}(U_{1})\leftarrow\Transpose{\text{Low}(U_{1})},
\hspace{15pt}
U_{2} \leftarrow U_{1} + P_{4},
\hspace{15pt}
U_{4} \leftarrow U_{2} + P_{3},\)
\State \Comment{1 half addition ($U_{5}=U_{1}+P_{4}+\Transpose{P_{4}}$ is symmetric):}
\State \(\text{Low}(U_{5}) \leftarrow \text{Low}(U_{2}) + \text{Low}(\Transpose{P_{4}}).\)

\State \Return{$\begin{smatrix} \text{Low}(U_{3}) &\\ U_{4} & \text{Low}(U_{5}) \end{smatrix}$.}
\end{algorithmic}
\end{spacing}
\end{algorithm}

\begin{restatable}[Appendix~\ref{thm:proof}]{proposition}{wishartp}\label{thm:wishartp}
Algorithm~\ref{alg:wishartp} is correct for any skew-orthogonal matrix~$Y$.
\end{restatable}
\subsection{Skew orthogonal matrices}\label{ssec:skeworthmat}
Algorithm~\ref{alg:wishartp} requires a skew-orthogonal matrix.
Unfortunately there are no skew-orthogonal matrices over~$\RR$, nor~$\QQ$.
Hence, we report no improvement in these cases.
In other domains, the simplest skew-orthogonal matrices just use a square root of~$-1$.
\subsubsection{Over the complex field}
Therefore Algorithm~\ref{alg:wishartp} is directly usable over~$\CC^{n{\times}n}$ with~${Y=i\,\IdentityMatrixOfSize{n}\in\CC^{n{\times}n}}$.
Further, usually, complex numbers are emulated by a pair of floats so then the multiplications by~${Y=i\,\IdentityMatrixOfSize{n}}$ are essentially free since they just exchange the real and imaginary parts, with one sign flipping.
Even though over the complex the product \textsc{zherk} of a matrix by its \emph{conjugate} transpose is more widely used, \textsc{zsyrk} has some applications, see for instance~\cite{Baboulin:2005:csyrk}.
\subsubsection{Negative one is a square}
Over some fields with prime characteristic, square roots of~$-1$ can be elements of the base field, denoted~$i$ in~$\F$ again.
There, Algorithm~\ref{alg:wishartp} only requires some
pre-multiplications by this square root (with
also~${Y=i\,\IdentityMatrixOfSize{n}\in\F^{n{\times}n}}$), but
\emph{within the field}.
\cref{lem:lemsqrt} thereafter characterizes these fields.
\begin{proposition}\label{lem:lemsqrt}
Fields with characteristic two, or with an odd
characteristic~${{p}\equiv{1}\bmod{4}}$, or finite fields that are an even extension, contain a square root of~$-1$.
\end{proposition}
\begin{proof}
  If~${p=2}$, then~${1=1^{2}=-1}$.
  If~${{{p}\equiv{1}}\bmod{4}}$, then half of the non-zero elements~$x$ in the base field of size~$p$ satisfy~${x^{\frac{p-1}{4}} \neq \pm 1}$ and then the square of the latter must be~$-1$.
  If the finite field~$\F$ is of cardinality~$p^{2k}$, then, similarly, there exists elements~${x^{\frac{p^{k}-1}{2}\frac{p^{k}+1}{2}}}$ different from~$\pm 1$ and then the square of the latter must be~$-1$.
\end{proof}
\subsubsection{Any field with prime characteristic}
Finally, we show that Algorithm~\ref{alg:wishartp} can also be run without any field extension, even when~$-1$ is not a square:
form the skew-orthogonal matrices constructed in
\cref{lem:pigeonhole}, thereafter, and use them directly as
long as the dimension of~$Y$ is even.
Whenever this dimension is odd, it is always possible to pad with zeroes so that~${\MatrixProduct{A}{\Transpose{A}}=\MatrixProduct{\begin{smatrix}A&0\end{smatrix}}{\begin{smatrix} \Transpose{A} \\ 0\end{smatrix}}}$.
\begin{proposition}\label{lem:pigeonhole}
Let~$\F$ be a field of characteristic~$p$, there exists~${(a,b)}$ in~${\F^{2}}$ such that the matrix:
\begin{equation}
\begin{smatrix}
a & b\\
-b & a
\end{smatrix}\tensorproduct{\IdentityMatrixOfSize{n}} =
\begin{smatrix}
a\, \IdentityMatrixOfSize{n} & b\, \IdentityMatrixOfSize{n}\\
-b\, \IdentityMatrixOfSize{n} & a\, \IdentityMatrixOfSize{n}
\end{smatrix}\quad \textrm{in}\quad \F^{2n{\times}2n}
\end{equation}
is skew-orthogonal.
\end{proposition}

\begin{proof}
Using the relation
\begin{equation}
\begin{smatrix}
a \,\IdentityMatrixOfSize{n} & b \,\IdentityMatrixOfSize{n}\\
-b \,\IdentityMatrixOfSize{n} & a \,\IdentityMatrixOfSize{n}
\end{smatrix}
\Transpose{%
\begin{smatrix}
a \,\IdentityMatrixOfSize{n} & b \,\IdentityMatrixOfSize{n}\\
-b \,\IdentityMatrixOfSize{n} & a \,\IdentityMatrixOfSize{n}
\end{smatrix}}=
(a^2+b^2)\,\IdentityMatrixOfSize{2n},
\end{equation}
it suffices to prove that there exist~$a,b$ such that~${a^{2}+b^{2}=-1}$.
In characteristic~2,~${{a=1},{b=0}}$ is a solution as~${1^{2}+0^{2}=-1}$.
In odd characteristic, there are~${\frac{p+1}{2}}$ distinct square elements~${x_{i}}^{2}$ in the base prime field.
Therefore, there are~$\frac{p+1}{2}$ distinct elements~${-1-{x_{i}}^{2}}$.
But there are only~$p$ distinct elements in the base field, thus there exists a couple~$(i,j)$ such
that~${-1-{x_{i}}^{2}}={x_{j}}^{2}$~\cite[Lemma~6]{Seroussi:1980:BBgfp}.
\end{proof}
\cref{lem:pigeonhole} shows that skew-orthogonal matrices
do exist for any field with prime characteristic.
For Algorithm~\ref{alg:wishartp}, we need to build them mostly for~${{p}\equiv{3}\bmod 4}$ (otherwise use \cref{lem:lemsqrt}).
\par
For this, without the extended Riemann hypothesis (\textsc{erh}), it is possible to use the decomposition of primes into squares:
\begin{enumerate}
\item Compute first a prime~${r=4pk+(3-1)p-1}$, then the relations~${{r}\equiv{1}\bmod{4}}$ and~${r\equiv{-1}\bmod{p}}$ hold;
\item Thus, results of~\cite{brillhart:1972:twosquares} allow one to decompose primes into squares and give a couple~${(a,b)}$ in~${\Z^{2}}$ such that~${a^2+b^2=r}$.
Finally, we get~${{{a^{2}+b^{2}}\equiv{-1}}\bmod{p}}$.
\end{enumerate}
By the prime number theorem the first step is polynomial in~$\log(p)$,
as is the second step (square root modulo a prime, denoted \FFSqrtname, has a cost close to exponentiation and then the rest of Brillhart's algorithm is \textsc{gcd}-like).
In practice, though, it is faster to use the following
Algorithm~\ref{alg:sosmodp}, even though the latter has a better
asymptotic complexity bound only if the \textsc{erh} is true.
\begin{algorithm}[htbp]\caption{\texttt{SoS}: Sum of squares decomposition over a finite field}\label{alg:sosmodp}
\begin{algorithmic}[1]
 \Require{${p\in\Primes\backslash\{2\}}$,~${k\in\Z}$.}
 \Ensure{${(a,b)\in\Z^{2}}$, s.t.~${a^{2}+b^{2}\equiv{k}\bmod{p}}$.}
 \If{$\left(\frac{k}{p}\right)=1$}
    \Comment{$k$ is a square mod~$p$}
    \State \Return{$\left(\FFSqrt{p}{k},0\right)$.}
 \Else \Comment{Find smallest quadratic non-residue}
 \State $s\assign 2$;
 \InlineWhile{$\left(\frac{s}{p}\right)==1$}{$s\assign s+1$}
 \EndIf
 \State ${c \assign \FFSqrt{p}{s-1}}$ \hfill\Comment{${s-1}$ must be a square}
 \State $r \assign k s^{-1} \bmod{p}$
 \State ${a \assign \FFSqrt{p}{r}}$ \Comment{Now~${{k}\equiv{a^{2}s}\equiv{a^{2}(1+c^{2})}\bmod{p}}$}
 \State \Return $\left(a, ac\bmod{p}\right)$
\end{algorithmic}
\end{algorithm}

\begin{proposition}\label{thm:sosmodpcorrect}
Algorithm~\ref{alg:sosmodp} is correct and, under the~\textsc{erh}, runs in expected time~${\widetilde{O}\bigl({\log}^{3}(p)\bigr)}$.
\end{proposition}
\begin{proof}
If~$k$ is square then the square of one of its square roots added to the square of zero is a solution.
Otherwise, the lowest quadratic non-residue (\textsc{lqnr}) modulo~$p$ is one plus a square~$b^{2}$ ($1$ is always a square so the \textsc{lqnr} is larger than~$2$).
For any generator of~$\Z_{p}$, quadratic non-residues, as well as their inverses ($s$ is invertible as it is non-zero and~$p$ is prime), have an odd discrete logarithm.
Therefore the multiplication of~$k$ and the inverse of the \textsc{lqnr} must be a square~$a^{2}$.
This means that the relation~${k=a^{2}\bigr(1+b^{2}\bigl)=a^{2}+{(ab)}^{2}}$ holds.
Now for the running time, under \textsc{erh}, the \textsc{lqnr} should be lower than~${3\log^{2}(p)/2-44\log(p)/5+13}$~\cite[Theorem~6.35]{Wedeniwski:2001:lqnr}.
The expected number of Legendre symbol computations is~$O\bigr(\log^{2}(p)\bigl)$ and this dominates the modular square root computations.
\end{proof}
\begin{remark}
Another possibility is to use randomization: instead of using the
lowest quadratic non-residue (\textsc{lqnr}), randomly select a
non-residue~$s$, and then decrement
it until~${s-1}$ is a quadratic residue ($1$ is a square so this will
terminate)\footnote{In practice, the running time seems very close
  to that of Algorithm~\ref{alg:sosmodp} anyway, see, e.g.\ the
  implementation in Givaro rev.~7bdefe6,
  \url{https://github.com/linbox-team/givaro}.}.
Also, when computing~$t$ sum of squares modulo the same prime, one can
compute the \textsc{lqnr}
\emph{only once} to get all the sum of squares with an expected cost bounded
by~${\widetilde{O}\bigl({{\log^{3}}(p)+t{\log^{2}}(p)\bigr)}}$.
\end{remark}
\begin{remark}\label{alg:FFSoS}
Except in characteristic~$2$ or in algebraic closures, where every element is a square anyway, Algorithm~\ref{alg:sosmodp} is easily extended over any finite field: compute the \textsc{lqnr} in the base prime field, then use Tonelli-Shanks or Cipolla-Lehmer algorithm to compute square roots in the extension field.

Denote by~$\FFSoS{q}{k}$ this algorithm decomposing~$k$ as a sum of squares within any finite field~$\F_{q}$.
This is not always possible over infinite fields, but there Algorithm~\ref{alg:sosmodp} still works anyway for the special case~${k=-1}$: just run it in the prime subfield.
\end{remark}
\subsection{Conjugate transpose}\label{ssec:herk}
Note that Algorithm~\ref{alg:wishartp} remains valid if transposition is replaced by \emph{conjugate transposition}, provided that there exists a matrix~$Y$ such that~${\MatrixProduct{Y}{\CTranspose{Y}}=-\IdentityMatrix{}}$.
This is not possible anymore over the complex field, but works for any even extension field, thanks to Algorithm~\ref{alg:sosmodp}:
if~$-1$ is a square in~$\F_{q}$,
then~${Y=\sqrt{-1}\cdot\IdentityMatrixOfSize{n}}$ still works;
otherwise there exists a square root~$i$ of~$-1$ in~$\F_{q^{2}}$, from
\cref{lem:lemsqrt}.
In the latter case, thus build~$(a,b)$, both in~$\F_{q}$, such
that~${a^{2}+b^{2}=-1}$.
Now~${Y=(a+ib)\cdot{}{\IdentityMatrixOfSize{n}}}$
in~${{\F_{q^{2}}}^{n{\times}n}}$ is appropriate: indeed, since~${{q}\equiv{{3}\bmod{4}}}$, we have that~${\overline{a+ib}={(a+ib)}^{q}={a-ib}}$.
\section{Analysis and implementation}
\subsection{Complexity bounds}
\begin{thm}\label{thm:complexitybound}
\cref{alg:wishartp}
requires~${\frac{2}{2^{\omega}-3}C_{\omega}n^{\omega} +\LO{n^{\omega}}}$ field operations, over
\CC~or over any field with prime characteristic.
\end{thm}
\begin{proof}
\cref{alg:wishartp} is applied recursively to compute three products~$P_{1}, P_{2}$ and~$P_{7}$, while~$P_{4}$ and~$P_{5}$ are computed in~$\textrm{MM}_\omega(n)=C_{\omega}n^{\omega}+\LO{n^{\omega}}$ using a general matrix multiplication algorithm.
We will show that applying the skew-orthogonal matrix~$Y$ to a~${{n}\times{n}}$ matrix costs~$yn^{2}$ for some constant~$y$ depending on the base field.
Then applying \cref{rq:7.5} thereafter, the cost~$T(n)$ of \cref{alg:wishartp} satisfies:
\begin{equation}\label{eq:complexity}
T(n) \leq 3T(n/2) + 2C_{\omega} {(n/2)}^{\omega} + (7.5+2y){(n/2)}^{2} + \LO{n^2}
\end{equation}
and~$T(4)$ is a constant.
Thus, by the master Theorem:
\begin{equation}
	T(n) \leq \frac{2C_{\omega}}{2^{\omega}-3}n^{\omega} +\LO{n^{\omega}} =
	\frac{2}{2^{\omega}-3}\textrm{MM}_{\omega}(n)+\LO{n^{\omega}}.
\end{equation}
\par
If the field is~$\mathbb{C}$ or satisfies the conditions of \cref{lem:lemsqrt}, there is a square root~$i$ of~$-1$.
Setting~${Y=i\,\IdentityMatrixOfSize{n/2}}$ yields~${y=1}$.
Otherwise, in characteristic~${p\equiv{3}\bmod{4}}$, \cref{lem:pigeonhole} produces~$Y$ equal to~${\begin{smatrix}a&b\\-b&a\end{smatrix}\tensorproduct{\IdentityMatrixOfSize{n/2}}}$ for which~${y=3}$.
As a subcase, the latter can be improved when~${p\equiv{3}\bmod{8}}$:
then~$-2$ is a square
(indeed,
$\left(\frac{-2}{p}\right)=\left(\frac{-1}{p}\right)\left(\frac{2}{p}\right)=(-1)^{\frac{p-1}{2}}(-1)^{\frac{p^2-1}{8}}=(-1)(-1)=1$).
Therefore, in this case set~${a=1}$ and~${b\equiv\sqrt{-2}\bmod{p}}$ such that the relation~${a^{2}+b^{2}=-1}$ yields~${Y=\begin{smatrix} 1 & \sqrt{-2}\\ -\sqrt{-2} & 1\end{smatrix}\tensorproduct{\IdentityMatrixOfSize{n/2}}}$ for which~${y=2}$.
\end{proof}
To our knowledge, the best previously known result was with a~$\frac{2}{2^{\omega}-4}$ factor instead, see e.g.~\cite[\S~6.3.1]{jgd:2008:toms}.
Table~\ref{tab:ffcomplex} summarizes the arithmetic complexity bound improvements.
\begin{table}[htbp]\centering\small%
\begin{tabular}{crrrr}
\toprule
Problem & Alg.\ & $\GO{n^{3}}$ & $\GO{n^{\log_2(7)}}$ & $\GO{n^{\omega}}$ \\
\midrule
\multirow{2}{*}{$\MatrixProduct{A}{\Transpose{A}} \in\F^{n{\times}n}$}
& \cite{jgd:2008:toms} & $n^{3}$ & $\frac{2}{3}\,\textrm{MM}_{\log_{2}(7)}(n)$& $\frac{2}{2^{\omega}-4}\,\textrm{MM}_{\omega}(n)$\\
& Alg.~\ref{alg:wishartp}  & $0.8 n^{3}$ & $\frac{1}{2}\,\textrm{MM}_{\log_{2}(7)}(n)$& $\frac{2}{2^{\omega}-3}\,\textrm{MM}_{\omega}(n)$ \\
\bottomrule
\end{tabular}
\caption{Arithmetic complexity bounds leading terms.}\label{tab:ffcomplex}\mvspace{-15pt}
\end{table}

Alternatively, over~$\mathbb{C}$, the~$3M$ method (Karatsuba) for
non-symmetric matrix multiplication reduces the number of
multiplications of real matrices from~$4$
to~$3$~\cite{Higham:1992:complex3M}:
if~$RR_{\omega}(n)$ is the cost of multiplying~${{n}\times{n}}$ matrices over~$\mathbb{R}$, then the~$3M$~method costs~$3RR_{\omega}(n)+\LO{n^\omega}$ operations over~$\mathbb{R}$.
Adapting this approach to the symmetric case yields a~$2M$ method to compute the product of a complex matrix by its transpose, using only~$2$ real products:
${H=\MatrixProduct{A}{\Transpose{B}}}$
and~${G=\MatrixProduct{(A+B)}{(\Transpose{A}-\Transpose{B})}}$.
Combining those into~$(G-\Transpose{H}+H)+i(H+\Transpose{H})$,
yields the product~${\MatrixProduct{(A+iB)}{(\Transpose{A}+i\Transpose{B})}}$.
This approach costs~${2RR_{\omega}}+\LO{n^{\omega}}$ operations in~$\mathbb{R}$.
\par
Classical algorithm~\cite[\S~6.3.1]{jgd:2008:toms} applies a divide and conquer approach directly on the complex field.
This would use only the equivalent of~$\frac{2}{2^{\omega}-4}$ complex floating point~$\matrixsize{n}{n}$ products.
Using the~$3M$ method for the complex products, this algorithm
uses overall~${\frac{6}{2^{\omega}-4}RR_{\omega}+\LO{n^\omega}}$ operations in~$\mathbb{R}$.
Finally, Algorithm~\ref{alg:wishartp} only costs~$\frac{2}{2^{\omega}-3}$ complex multiplications for a leading term bounded by~$\frac{6}{2^{\omega}-3}\textrm{RR}_{\omega}$, better than~$2\textrm{RR}_\omega$ for~${\omega>\log_{2}(6)\approx 2.585}$.
This is summarized in Table~\ref{tab:complexcomplex}, replacing~$\omega$ by~$3$ or~$\log_{2}(7)$.
\begin{table}[htbp]\centering\small%
\begin{tabular}{crrrr}
\toprule
Problem & Alg.\ &  $\textrm{MM}_3(n)$ & $\textrm{MM}_{\log_2 7} (n)$ & $\textrm{MM}_\omega(n)$ \\
\midrule
\multirow{2}{*}{$\MatrixProduct{A}{B} \in\CC^{n{\times}n}$}
& naive & $8n^{3}$ & $4\,\textrm{RR}_{\log_{2}(7)}(n)$& $4\,\textrm{RR}_{\omega}(n)$\\
& 3M    & $6n^{3}$ & $3\,\textrm{RR}_{\log_{2}(7)}(n)$& $3\,\textrm{RR}_{\omega}(n)$\\
\midrule
\multirow{3}{*}{$\MatrixProduct{A}{\Transpose{A}}\in\CC^{n{\times}n}$}
& 2M    & $4n^{3}$ & $2\,\textrm{RR}_{\log_{2}(7)}(n)$& $2\,\textrm{RR}_{\omega}(n)$\\
& \cite{jgd:2008:toms} & $3n^{3}$ & $2\,\textrm{RR}_{\log_{2}(7)}(n)$
& $\frac{6}{2^{\omega}-4}\,\textrm{RR}_{\omega}(n)$\\
& Alg.~\ref{alg:wishartp}  & $2.4 n^{3}$ & $\frac{3}{2}\,\textrm{RR}_{\log_{2}(7)}(n)$ & $\frac{6}{2^{\omega}-3}\,\textrm{RR}_{\omega}(n)$\\
\bottomrule
\end{tabular}
\caption{Symmetric multiplication over~$\mathbb{C}$: leading term of the cost in number of operations over~$\mathbb{R}$.}\label{tab:complexcomplex}\mvspace{-10pt}
\end{table}
\begin{remark}\label{rq:7.5}
Each recursive level of \cref{alg:wishartp} is composed of 9 block additions.
An exhaustive search on all symmetric algorithms derived from Strassen's showed that this number is minimal in this class of algorithms.
Note also that~$3$ out of these~$9$ additions in \cref{alg:wishartp} involve symmetric matrices and are therefore only performed on the lower triangular part of the matrix.
Overall, the number of scalar additions is~${6n^{2}+3/2n(n+1)=15/2n^{2}+1.5n}$, nearly half of the optimal in the non-symmetric case~\cite[Theorem~1]{bshouty:1995a}.
\end{remark}
To further reduce the number of additions, a promising approach is that undertaken in~\cite{Karstadt:2017:strassen,Beniamini:2019:fmmsd}.
This is however not clear to us how to adapt our strategy to their recursive transformation of basis.
\subsection{Implementation and scheduling}
This section reports on an implementation of Algorithm~\ref{alg:wishartp} over prime fields.
We propose in Table~\ref{tab:schedule:AAT} and Figure~\ref{fig:DAG:AAT} a schedule for the operation~${C\leftarrow\MatrixProduct{A}{\Transpose{A}}}$ using no more extra storage than the unused upper triangular part of the result~$C$.

\begin{table}[htb]
  \small
  \begin{center}
    \begin{tabular}{cll|cll}
      \toprule
      \# & operation & loc.\ & \# & operation & loc.\\
      \midrule
	1&$S_{1}=\MatrixProduct{(A_{21}-A_{11})}{Y}$	&$C_{21}$&9&$U_{1}=P_{1}+P_{5}$&$C_{12}$\\
	2&$S_{2}=A_{22}-\MatrixProduct{A_{21}}{Y}$	&$C_{12}$&&$\text{Up}(U_{1})=\Transpose{\text{Low}(U_{1})}$&$C_{12}$\\
	3&$\Transpose{P_{4}}=\MatrixProduct{S_{2}}{\Transpose{S_{1}}}$&$C_{22}$&10&$U_{2}=U_{1}+P_{4}$&$C_{12}$\\
	4&$S_{3}=S_{1}-A_{22}$&$C_{21}$&11&$U_{4}=U_{2}+P_{3}$&$C_{21}$\\
	5&$P_{5}=\MatrixProduct{S_{3}}{\Transpose{S_{3}}}$&$C_{12}$&12&$U_{5}=U_{2}+\Transpose{P_{4}}$&$C_{22}$\\
	6&$S_{4}=S_{3}+A_{12}$&$C_{11}$&13&$P_{2}=\MatrixProduct{A_{12}}{\Transpose{A_{12}}}$&$C_{12}$\\
	7&$P_{3}=\MatrixProduct{A_{22}}{\Transpose{S_{4}}}$&$C_{21}$&14&$U_{3}=P_{1}+P_{2}$&$C_{11}$\\
	8&$P_{1}=\MatrixProduct{A_{11}}{\Transpose{A_{11}}}$&$C_{11}$\\
      \bottomrule
    \end{tabular}
    \caption{%
Memory placement and schedule of tasks to compute the lower triangular part of~${C\leftarrow \MatrixProduct{A}{\Transpose{A}}}$ when~${k\leq n}$.
The block~$C_{12}$ of the output matrix is the only temporary used.
}\label{tab:schedule:AAT}\mvspace{-10pt}
  \end{center}
\end{table}
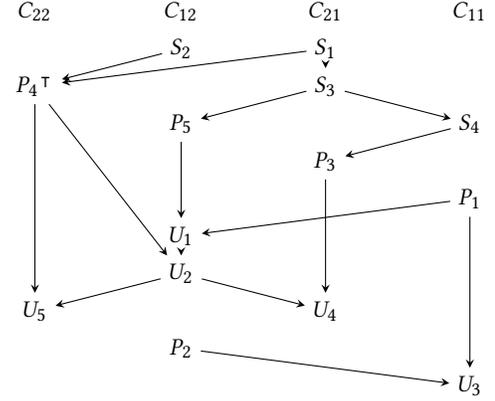
\begin{figure}[htb]
  \begin{center}
\begin{tikzpicture}%
  \matrix (m) [matrix of math nodes, row sep=.1em, column sep=4em ]
  {%
    C_{22} & C_{12} & C_{21} & C_{11} \\
          & S_{2}   & S_{1}   &       \\
     \Transpose{P_{4}}  &       & S_{3}   &       \\
          & P_{5}   &       &  S_{4}  \\
          &       & P_{3}   &       \\
          &       &       &  P_{1}  \\
          &  U_{1}  &       &       \\
          &  U_{2}  &       &       \\
   U_{5}    &       & U_{4}   &       \\
          &  P_{2}  &       &       \\
          &       &       &  U_{3}  \\
  };
  \path[-stealth]
  (m-2-2) edge (m-3-1)
  (m-2-3) edge (m-3-1)
          edge (m-3-3)
  (m-3-3) edge (m-4-2)
          edge (m-4-4)
  (m-4-4) edge (m-5-3)
  (m-4-2) edge (m-7-2)
  (m-6-4) edge (m-7-2)
  (m-7-2) edge (m-8-2)
  (m-3-1) edge (m-8-2)
          edge (m-9-1)
  (m-8-2) edge (m-9-1)
          edge (m-9-3)
  (m-5-3) edge (m-9-3)
  (m-10-2) edge (m-11-4)
  (m-6-4) edge (m-11-4);
\end{tikzpicture}\mvspace{-10pt}
\caption{\textsc{dag} of the tasks and their memory location for the computation of~${C\leftarrow \MatrixProduct{A}{\Transpose{A}}}$ presented in Table~\ref{tab:schedule:AAT}.}\label{fig:DAG:AAT}
\Description[]{\textsc{dag} of the tasks and their memory location for the computation of~${C\leftarrow \MatrixProduct{A}{\Transpose{A}}}$ presented in Table~\ref{tab:schedule:AAT}.}\mvspace{-10pt}
\end{center}
\end{figure}
\begin{table}[htb]
\small
\begin{center}
\begin{tabular}{ll|ll}
\toprule
 operation & loc.\ & operation & loc.\\
\midrule
${S_{1}=\MatrixProduct{(A_{21}-A_{11})}{Y}}$&tmp&$P_{1}=\alpha\MatrixProduct{A_{11}}{\Transpose{A_{11}}}$&tmp\\
${S_{2}=A_{22}-\MatrixProduct{A_{21}}{Y}}$&$C_{12}$&$U_{1}=P_{1}+P_{5}$&$C_{12}$\\
$\text{Up}(C_{11})=\Transpose{\text{Low}(C_{22})}$&$C_{11}$&$\text{Up}(U_{1})=\Transpose{\text{Low}(U_{1})}$&$C_{12}$\\
$\Transpose{P_{4}}=\alpha\MatrixProduct{S_{2}}{\Transpose{S_{1}}}$&$C_{22}$&$U_{2}=U_{1}+P_{4}$&$C_{12}$\\
$S_{3}=S_{1}-A_{22}$&tmp&$U_{4}=U_{2}+P_{3}$&$C_{21}$\\
$P_{5}=\alpha\MatrixProduct{S_{3}}{\Transpose{S_{3}}}$&$C_{12}$&$U_{5}=U_{2}+\Transpose{P_{4}}+\beta\Transpose{\text{Up}(C_{11})}$&$C_{22}$\\
$S_{4}=S_{3}+A_{12}$&tmp&$P_{2}=\alpha\MatrixProduct{A_{12}}{\Transpose{A_{12}}}+\beta C_{11}$&$C_{11}$\\
$P_{3}=\alpha\MatrixProduct{A_{22}}{\Transpose{S_{4}}}+\beta C_{21}$&$C_{21}$&$U_{3}=P_{1}+P_{2}$&$C_{11}$\\
\bottomrule
\end{tabular}
\caption{%
Memory placement and schedule of tasks to compute the lower triangular part of~${C\leftarrow \alpha \MatrixProduct{A}{\Transpose{A}}+\beta C}$ when~${k\leq n}$.
The block~$C_{12}$ of the output matrix as well as an~${n/2\times n/2}$ block tmp are used as temporary storages.}\label{tab:schedule:AATpC}\mvspace{-10pt}
\end{center}
\end{table}
\begin{figure}[htb]
  \begin{center}
\begin{tikzpicture}%
  \matrix (m) [matrix of math nodes, row sep=.1em, column sep=3em ]
  {%
   C_{11}  &C_{22} & C_{12} & \text{tmp} &C_{21} \\
\text{Up}(C_{11})& & S_{2}   & S_{1}   &        \\
          & \Transpose{P_{4}}&  & S_{3}   &    \\
          &       & P_{5}   & S_{4}   &       \\
          &       &       & P_{1}   &  P_{3}  \\
          &       &  U_{1}  &       &       \\
          &       &  U_{2}  &       &       \\
          &  U_{5}  &       &       &  U_{4}    \\
    P_{2}   &       &       &       &        \\
    U_{3}   &       &       &       &       \\
  };
  \path[-stealth]
  (m-1-2) edge (m-2-1)
  (m-1-1) edge (m-2-1)
  (m-2-3) edge (m-3-2)
  (m-2-4) edge (m-3-2)
          edge (m-3-4)
  (m-3-4) edge (m-4-3)
          edge (m-4-4)
  (m-4-4) edge (m-5-5) %
  (m-4-3) edge (m-6-3) %
  (m-6-3) edge (m-7-3) %
  (m-3-2) edge (m-7-3) %
          edge (m-8-2) %
  (m-5-4) edge (m-6-3) %
  (m-5-5) edge (m-8-5) %
  (m-7-3) edge (m-8-5)  %
          edge (m-8-2) %
  (m-5-4) edge[bend left] (m-10-1) %
  (m-9-1) edge (m-10-1) %
  (m-1-5) edge (m-5-5) %
  (m-2-1) edge (m-8-2) %
          edge (m-9-1); %
\end{tikzpicture}\mvspace{-10pt}
\caption{\textsc{dag} of the tasks and their memory location for the computation of~${C\leftarrow \alpha \MatrixProduct{A}{\Transpose{A}} + \beta C}$ presented in Table~\ref{tab:schedule:AATpC}.}\label{fig:DAG:AATpC}
\Description[]{\textsc{dag} of the tasks and their memory location for the computation of~${C\leftarrow \alpha \MatrixProduct{A}{\Transpose{A}} + \beta C}$ presented in Table~\ref{tab:schedule:AATpC}.}\mvspace{-10pt}
  \end{center}
  \end{figure}
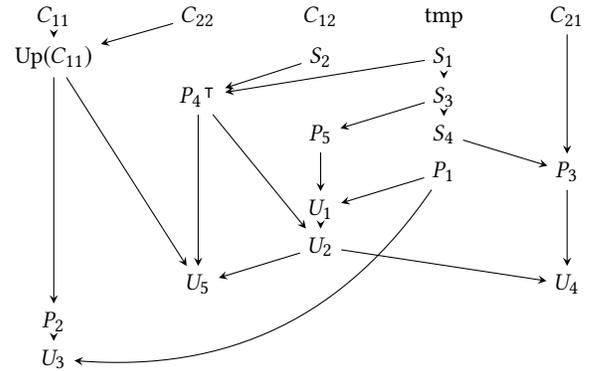
For the more general operation~${C\leftarrow \alpha \MatrixProduct{A}{\Transpose{A}} + \beta C}$, Table~\ref{tab:schedule:AATpC} and Figure~\ref{fig:DAG:AATpC} propose a schedule requiring only an additional~${{n/2}\times{n/2}}$ temporary storage.
These algorithms have been implemented as the \texttt{fsyrk} routine
in the \texttt{fflas-ffpack} library for dense linear algebra over a
finite
field~\cite[\href{https://github.com/linbox-team/fflas-ffpack/commit/0a91d61e6518568b006873076df925fcd6fcc112}{from
  commit 0a91d61e}]{fflas19}.

Figure~\ref{fig:perfs} compares the computation speed in effective
Gfops (defined as~${n^{3}/(10^{9}\times\textrm{time})}$) of this
implementation over~${\Z/131071\Z}$ with that of the double precision
\textsc{blas} routines \texttt{dsyrk}, the classical cubic-time routine
over a finite field (calling \texttt{dsyrk} and performing modular
reductions on the result), and the classical divide and conquer
algorithm~\cite[\S~6.3.1]{jgd:2008:toms}.
\begin{figure}[htb]
  \begin{center}
    \includegraphics[width=.46\textwidth]{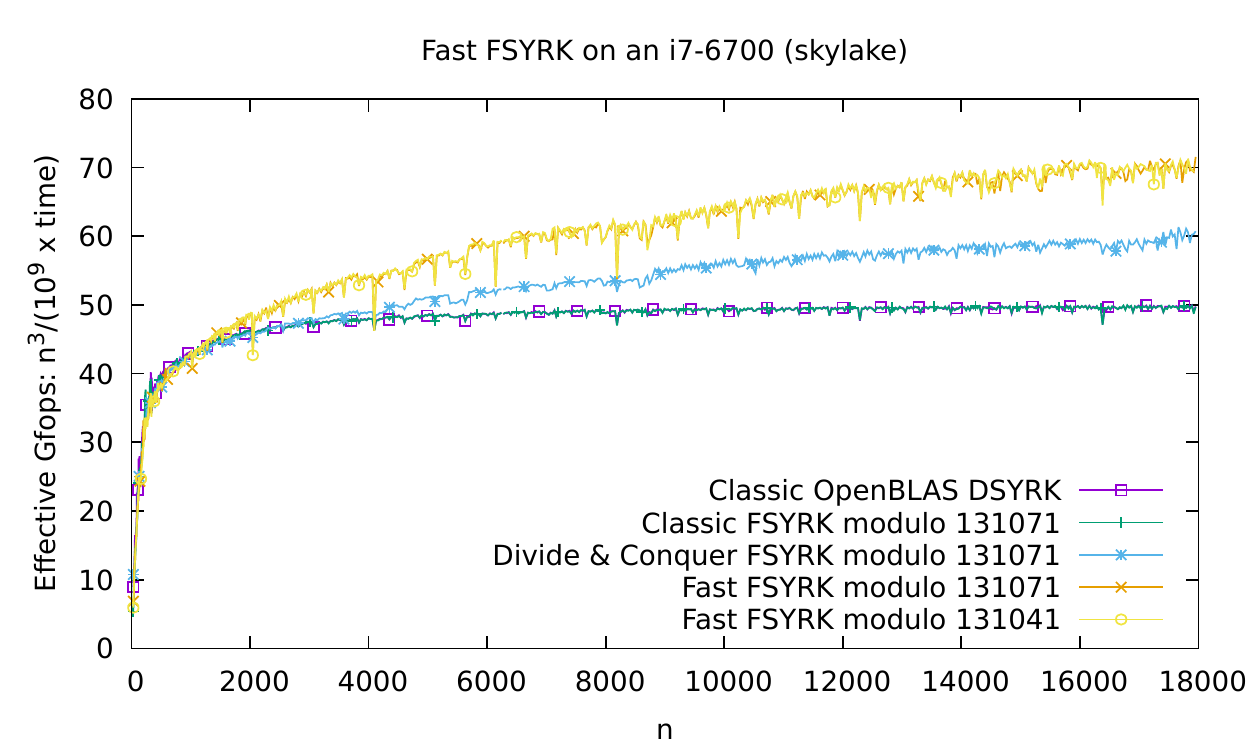}\mvspace{-10pt}
    \caption{Speed of an implementation of Algorithm~\ref{alg:wishartp}}\label{fig:perfs}%
\mvspace{-10pt}
\Description[Speed of an implementation of Algorithm~\ref{alg:wishartp}]{Speed of an implementation of Algorithm~\ref{alg:wishartp}}
  \end{center}
\end{figure}
The \texttt{fflas-ffpack} library is linked with
Open\textsc{blas}~\cite[v0.3.6]{openblas}
and compiled with \texttt{gcc-9.2} on an Intel skylake~i7-6700 running
a Debian \textsc{gnu}/Linux system (v5.2.17).
\par

The slight overhead of performing the modular reductions is quickly compensated by the speed-up of the sub-cubic algorithm (the threshold for a first recursive call is near~${n=2000}$).
The classical divide and conquer approach also speeds up the classical algorithm, but starting from a larger threshold, and hence at a slower pace.
Lastly, the speed is merely identical modulo~${131041}$, where square roots of~$-1$ exist, thus showing the limited overhead of the preconditioning by the matrix~$Y$.
\section{\syrk with block diagonal scaling}
Symmetric rank k updates are a key building block for symmetric triangular factorization
algorithms, for their efficiency is one of the bottlenecks.
In the most general setting (indefinite factorization), a block diagonal scaling by a matrix~$D$, with~$1$ or~$2$ dimensional diagonal blocks, has to be inserted within the product, leading to the operation:~${{C}\leftarrow{C -\MatrixProduct{A}{\MatrixProduct{D}{\Transpose{A}}}}}$.
\par
Handling the block diagonal structure over the course of the recursive
algorithm may become tedious and quite expensive.
For instance, a~$\matrixsize{2}{2}$ diagonal block might have to be cut by a recursive split.
We will see also in the following that non-squares in the diagonal need to be dealt with in pairs.
In both cases it might be necessary to add a column to deal with these cases: this is potentially~$\GO{\log_{2}(n)}$ extra columns in a recursive setting.
\par
Over a finite field, though, we will show in this section, how to factor the block-diagonal matrix~$D$ into~${D=\MatrixProduct{\Delta}{\Transpose{\Delta}}}$, \emph{without needing any field extension}, and then compute instead~${\MatrixProduct{(\MatrixProduct{A}{\Delta})}{\Transpose{(\MatrixProduct{A}{\Delta})}}}$.
\cref{alg:ApsiAT}, deals with non-squares and~$\matrixsize{2}{2}$ blocks only once beforehand, introducing no more than~$2$ extra-columns overall.
\cref{ssec:nonsquares} shows how to factor a diagonal matrix, without resorting
to field extensions for non-squares.
Then \cref{ssec:antiodd,ssec:antieven} show how to deal with the~${{2}\times{2}}$ blocks depending on the characteristic.
\subsection{Factoring non-squares within a finite field}\label{ssec:nonsquares}
First we give an algorithm handling pairs of non-quadratic residues.\\
\begin{algorithm}[htbp]\caption{: \texttt{nrsyf}: Sym.\ factorization. of a pair of
    non-residues}\label{alg:factnonres}
\begin{algorithmic}[1]
\Require ${(\alpha,\beta)\in{\F_{q}}^{2}}$, both being quadratic non-residues.
\Ensure ${Y\in{\F_{q}}^{\matrixsize{2}{2}}}$, s.t.~${\MatrixProduct{Y}{\Transpose{Y}}=\begin{smatrix}\alpha&0\\0&\beta\end{smatrix}}$.
\State ${{(a,b)}\assign{\FFSoS{q}{\alpha}}}$;
\hfill\Comment{${\alpha=a^{2}+b^{2}}$}
\State ${{d}\assign{a\,\FFSqrt{q}{\beta\alpha^{-1}}}}$; \hfill\Comment{${d^{2}=a^{2}\beta\alpha^{-1}}$}
\State ${{c}\assign{-bda^{-1}}}$; \hfill\Comment{${ac+bd=0}$}
\State \Return ${Y\assign \begin{smatrix}a&b\\c&d\end{smatrix}}$.
\end{algorithmic}
\end{algorithm}
\begin{proposition}\label{thm:factnonrescorrect}
 Algorithm~\ref{alg:factnonres} is correct.
 \end{proposition}
\begin{proof}
Given~$\alpha$ and~$\beta$ quadratic non-residues, the couple~$(a,b)$, such that~${\alpha=a^{2}+b^{2}}$, is found by the algorithm of Remark~\ref{alg:FFSoS}.
Second, as~$\alpha$ and~$\beta$ are quadratic non-residues, over a
finite field their quotient is a residue since:~${{\left(\beta\alpha^{-1}\right)}^{\frac{q-1}{2}}=\frac{-1}{-1}=1}$.
Third, if~${c}$ denotes~${-bda^{-1}}$ then~${c^{2}+d^{2}}$ is equal
to~${{(-bd/a)}^{2}+d^{2}}$ and thus to~${(b^{2}/a^{2}+1)d^{2}}$; this
last quantity is equal to~${(\alpha)d^{2}/a^{2}}$ and then to~${\alpha{(a\sqrt{\beta/\alpha})}^{2}/a^{2}=\alpha(a^{2}\beta/\alpha)/a^{2}=\beta}$.
Fourth,~$a$ (or w.l.o.g.~$b$) is invertible.
Indeed,~$\alpha$ is not a square, therefore it is non-zero and thus one of~$a$ or~$b$ must be non-zero.
Finally, we obtain the cancellation~${ac+bd = a(-dba^{-1})+bd =
  -db+bd=0}$ and the matrix product~$\MatrixProduct{Y}{\Transpose{Y}}$
is~${\begin{smatrix}a&b\\c&d\end{smatrix}\begin{smatrix}a&c\\b&d\end{smatrix}=\begin{smatrix}a^{2}+b^{2}&ac+bd\\ac+bd&c^{2}+d^{2}\end{smatrix}=\begin{smatrix}\alpha&0\\0&\beta\end{smatrix}}$.
\end{proof}

Using Algorithm~\ref{alg:factnonres}, one can then factor any diagonal matrix within a finite field as a symmetric product with a tridiagonal matrix.
This can then be used to compute efficiently~$\MatrixProduct{A}{\MatrixProduct{D}{\Transpose{A}}}$ with~$D$ a diagonal matrix: factor~$D$ with a tridiagonal matrix~${D=\MatrixProduct{\Delta}{\Transpose{\Delta}}}$, then pre-multiply~$A$ by this tridiagonal matrix and run a fast symmetric product on the resulting matrix.
This is shown in Algorithm~\ref{alg:diagADAT}, where the overhead,
compared to simple matrix multiplication, is only~$\GO{n^{2}}$ (that is~$\GO{n}$ square roots and~$\GO{n}$ column scalings).
\begin{algorithm}[htbp]
  \caption{\texttt{syrkd}: sym.
    matrix product
    with diagonal scaling}
    \label{alg:diagADAT}
\begin{algorithmic}[1]
 \Require{${A\in{\F_{q}}^{\matrixsize{m}{n}}}$
   and~${D=\text{Diag}(d_{1},\ldots,d_n)\in{\F_{q}}^{\!n\times n}}$}
 \Ensure{$\MatrixProduct{A}{\MatrixProduct{D}{\Transpose{A}}}$
   in~${\F_{q}}^{\!\matrixsize{m}{m}}$}

	\If{number of quadratic non-residues in~${\{d_1,\ldots,d_n\}}$ is odd}
   \ Let~$d_{\ell}$ be one of the quadratic non-residues
   \State $\bar{D} \assign \text{\text{Diag}}(d_{1},\ldots,d_{n},d_{\ell})\in{\F_{q}}^{\!\matrixsize{(n+1)}{(n+1)}}$
   \State $\bar{A} \assign \begin{smatrix}A&0\end{smatrix}\in{\F_{q}}^{{m}{\times}{(n+1)}}$
     \Comment{Augment~$A$ with a zero column}
 \Else
 \State $\bar{D} \assign \text{\text{Diag}}(d_{1},\ldots,d_n)\in{\F_{q}}^{n{\times}n}$
 \State $\bar{A} \assign A\in{\F_{q}}^{m{\times}n}$
 \EndIf
 \ForAll{quadratic residues~$d_{j}$ in~$\bar{D}$}
   \State $\bar{A}_{*,j} \assign \FFSqrt{q}{d_{j}}\cdot{}\bar{A}_{*,j}$
     \Comment{Scale col.~$j$ of~$\bar{A}$ by a sq.\ root of~$d_{j}$}
 \EndFor
 \ForAll{distinct pairs of quadratic non-residues~$(d_{i},d_{j})$ in~$\bar{D}$}
 \State $\Delta \assign \texttt{nrsyf}(d_i,d_j)$
 \Comment{$\MatrixProduct{\Delta}{\Transpose{\Delta}}=\begin{smatrix}d_{i}&0\\0&d_{j}\end{smatrix}$
   using~\cref{alg:factnonres}}
 \State $\begin{smatrix}\bar{A}_{*,i}&\bar{A}_{*,j}\end{smatrix}
 \assign \MatrixProduct{\begin{smatrix}\bar{A}_{*,i}&\bar{A}_{*,j}\end{smatrix}}{\Delta}
   $;
 \EndFor
 \State\Return $\texttt{syrk}(\bar{A})$ \Comment{$\MatrixProduct{\bar{A}}{\Transpose{\bar{A}}}$ using~\cref{alg:wishartp}}
\end{algorithmic}
\end{algorithm}
\subsection{Antidiagonal and antitriangular blocks}\label{ssec:antianti}
In general, an~$\MatrixProduct{L}{\MatrixProduct{D}{\Transpose{L}}}$ factorization may have antitriangular or antidiagonal blocks in~$D$~\cite{Dumas:2018:ldlt}.
In order to reduce to a routine for fast symmetric multiplication with diagonal scaling, these blocks
need to be processed once for all, which is what this section is about.
\subsubsection{Antidiagonal blocks in odd characteristic}\label{ssec:antiodd}
In odd characteristic, the~$2$-dimensional blocks in an~$\MatrixProduct{L}{\MatrixProduct{D}{\Transpose{L}}}$ factorization are only of the
form~$\begin{smatrix}0&\beta\\\beta&0\end{smatrix}$, and always have  the symmetric factorization:
\begin{equation}\label{eq:antidiagonal}
\begin{smatrix}1&1\\1&-1\end{smatrix}
\begin{smatrix}\frac{1}{2}\beta&0\\0&-\frac{1}{2}\beta\end{smatrix}
\Transpose{\begin{smatrix}1&1\\1&-1\end{smatrix}}
= \begin{smatrix}0&\beta\\\beta&0\end{smatrix}.
\end{equation}
This shows the reduction to the diagonal case (note the requirement that~$2$ is invertible).
\subsubsection{Antitriangular blocks in characteristic~2}\label{ssec:antieven}
In characteristic~2, some~${\matrixsize{2}{2}}$ blocks might not be reduced further than an antitriangular form:~$\begin{smatrix}0&\beta\\ \beta&\gamma\end{smatrix}$, with~${\gamma\neq 0}$.
\par
In characteristic~2 every element is a square, therefore those antitriangular blocks can be factored as shown in \cref{eq:antitriangular}:
\begin{equation}\label{eq:antitriangular} \arraycolsep=.55\arraycolsep
\begin{smatrix}0&\beta\\\beta&\gamma\end{smatrix} =
\left(\begin{smatrix}\beta\gamma^{-1/2}&0\\0&\gamma^{1/2}\end{smatrix}
\begin{smatrix}1&1\\1&0\end{smatrix}\right)
\Transpose{\left(\begin{smatrix}\beta\gamma^{-1/2}&0\\0&\gamma^{1/2}\end{smatrix}\begin{smatrix}1&1\\1&0\end{smatrix}\right)}.
\end{equation}
Therefore the antitriangular blocks also reduce to the diagonal case.%
\subsubsection{Antidiagonal blocks in characteristic~2}
The symmetric factorization in this case
might require an extra row or column~\cite{Lempel:1975:BBft}
as shown in \cref{eq:antidiagonaleven}:
\begin{equation}\label{eq:antidiagonaleven} \arraycolsep=.8\arraycolsep
\begin{smatrix}1&0\\0&\beta\end{smatrix}
\begin{smatrix}1&0&1\\0&1&1\end{smatrix}
\Transpose{\left(\begin{smatrix}1&0\\0&\beta\end{smatrix}\begin{smatrix}1&0&1\\0&1&1\end{smatrix}\right)}=\begin{smatrix}0&\beta\\\beta&0\end{smatrix}\bmod{2}.
\end{equation}
A first option is to augment~$A$ by one column for each antidiagonal block, by
applying the~$2{\times}3$ factor in \cref{eq:antidiagonaleven}.
However one can instead combine a diagonal element, say~$x$, and an antidiagonal
block as shown in \cref{eq:pentadiagonal}.
\begin{equation}\label{eq:pentadiagonal}
\begin{smatrix}\sqrt{x}&\sqrt{x}&\sqrt{x}\\1&0&1\\0&\beta&\beta\end{smatrix}
\Transpose{\begin{smatrix}\sqrt{x}&\sqrt{x}&\sqrt{x}\\1&0&1\\0&\beta&\beta\end{smatrix}}
=\begin{smatrix}x&0&0\\0&0&\beta\\0&\beta&0\end{smatrix}\bmod{2}.
\end{equation}
Hence, any antidiagonal block can be combined with any~${1{\times}1}$ block to
form a symmetric factorization.
\par
There remains the case when there are no~${1{\times}1}$ blocks.
Then, one can use \cref{eq:antidiagonaleven} once, on the first antidiagonal block, and add column to~$A$.
This indeed extracts the antidiagonal elements and creates
a~${3{\times}3}$ identity block in the middle. Any one of its three
ones can then be used as~$x$ in a further combination
with the next antidiagonal blocks.
\cref{alg:ApsiAT} sums up the use of
\cref{eq:antidiagonal,eq:antitriangular,eq:antidiagonaleven,eq:pentadiagonal}.
\begin{algorithm}[!ht]
  \caption{: \texttt{syrkbd}: sym.
    matrix product
    with block diag. scaling}\label{alg:ApsiAT}
\begin{algorithmic}[1]
\Require{$A\in{\F_{q}}^{\matrixsize{m}{n}}$;
  $B\in{\F_{q}}^{\matrixsize{n}{n}}$, block diagonal with scalar or
$2$-dimensional blocks of the form \(
\left(\begin{smallmatrix}  0&\beta\\ \beta&\gamma\end{smallmatrix}\right)
\) with $\beta\neq 0$}
\Ensure{$\MatrixProduct{A}{\MatrixProduct{B}{\Transpose{A}}} \in {\F_{q}}^{m\times m}$}
\State $\bar{A}\assign A\in{\F_{q}}^{m{\times}n}$;  $\bar{D}\assign \text{I}_{n}$
\State \InlineForAll{scalar blocks in~$B$ at position~$j$}{$\bar{D}_j\leftarrow B_{j,j}$}

\If{$q$ is odd} \Comment{Use Eq.~(\ref{eq:antidiagonal})}
 \ForAll{symmetric antidiagonal blocks in~$B$ at~${(j,j+1)}$}
    \State $\beta\assign B_{j,j+1}(=B_{j+1,j})$
    \State $\bar{D}_j\leftarrow \frac{1}{2}\beta$ ; $\bar{D}_{j+1}\leftarrow -\frac{1}{2}\beta$
    \State
    $\begin{smatrix}\bar{A}_{*,i}&\bar{A}_{*,j}\end{smatrix}\leftarrow\begin{smatrix}\bar{A}_{*,i}&\bar{A}_{*,j}\end{smatrix}\begin{smatrix}1&1\\1&-1\end{smatrix}$
\vspace{-4pt}
 \EndFor
 \Else
 \ForAll{antitriangular blocks in~$B$ at position~$(j,j+1)$}
   \State $\beta \assign B_{j,j+1}(=B_{j+1,j})$ ;
          $\delta \assign \FFSqrt{q}{B_{j+1,j+1}}$;
    \State $\bar{A}_{*,j}\assign \beta\delta^{-1}\cdot{}\bar{A}_{*,j}$
    \Comment{Scale column~$j$ of~$\bar{A}$}
    \State $\bar{A}_{*,j+1}\assign \delta\cdot{}\bar{A}_{*,j+1}$
      \Comment{Scale column~$j+1$ of~$\bar{A}$}
    \State $\bar{A}_{*,j+1}\assign\bar{A}_{*,j+1}+\bar{A}_{*,j}$
      \Comment{Use Eq.~(\ref{eq:antitriangular})}
    \State Swap columns~$j$ and~$j+1$ of~$\bar{A}$
 \EndFor
\vspace{-1pt}
 \If{there are~$n/2$ antidiagonal blocks in~$B$} \Comment{Use Eq.~(\ref{eq:antidiagonaleven})}
   \State $\beta\assign B_{1,2}(=B_{2,1})$
   \State $\bar{A}_{*,2}\assign \beta\cdot{}\bar{A}_{*,2}$ ;
          $\bar{A}\assign\begin{smatrix}\bar{A}&\bar{A}_{*,1}+\bar{A}_{*,2}\end{smatrix}\in{\F_{q}}^{m{\times}(n+1)}$
   \State $\ell \assign 1$ ; 
          $\delta \assign 1$
 \Else
   \State $\delta \assign \FFSqrt{q}{\bar{D}_{\ell,\ell}}$ where $\ell$ is s.t.  $\bar{D}_{\ell,\ell}$ is a scalar block
\EndIf
\vspace{-2pt}
\ForAll{remaining antidiagonal blocks in~$B$ at~${(j,j+1)}$}
   \State $\beta\assign B_{j,j+1}(=B_{j+1,j})$
     \Comment{Use Eq.~(\ref{eq:pentadiagonal})}
     \State $\bar{A}_{*,\ell}\assign \delta\cdot{}\bar{A}_{*,\ell}$ ;
            $\bar{A}_{*,j+1}\assign \beta\cdot{}\bar{A}_{*,j+1}$
\vspace{-1pt}
   \State $\begin{smatrix}\bar{A}_{*,\ell}&\bar{A}_{*,j}&\bar{A}_{*,j+1}\end{smatrix}\assign\MatrixProduct{\begin{smatrix}\bar{A}_{*,\ell}&\bar{A}_{*,j}&\bar{A}_{*,j+1}\end{smatrix}}{\begin{smatrix}1&1&1\\1&0&1\\0&1&1\end{smatrix}}$
\vspace{-1pt}
   \State ${\delta}\assign{1}$
 \EndFor
\EndIf
\State\Return $\texttt{syrkd}(\bar{A},\bar{D})$
\Comment{$\MatrixProduct{\bar{A}}{\MatrixProduct{\bar{D}}{\Transpose{\bar{A}}}}$ using~\cref{alg:diagADAT}}
\end{algorithmic}
\end{algorithm}

\bibliographystyle{abbrvurl}
\bibliography{strassen}
\newpage
\appendix
\section{Appendix}

\subsection{Proof of Proposition~\ref{thm:wishartp}}\label{thm:proof}
\wishartp*
\begin{proof}
If~$Y$ is skew-orthogonal, then~${\MatrixProduct{Y}{\Transpose{Y}}=-\IdentityMatrix{}}$.
First,
\begin{equation}
\label{eq:prf:u3}
	U_{3} = P_{1}+P_{2} = \MatrixProduct{A_{11}}{\Transpose{A_{11}}}
	+\MatrixProduct{A_{12}}{\Transpose{A_{12}}} = C_{11}.
\end{equation}
Denote by~$R_{1}$ the product:
\begin{equation}
\label{eq:prf:r1}
\begin{split}
R_{1} &= \MatrixProduct{A_{11}}{\MatrixProduct{Y}{\Transpose{S_{2}}}}
= \MatrixProduct{A_{11}}{\MatrixProduct{Y}{(\Transpose{A_{22}} -
\MatrixProduct{\Transpose{Y}}{\Transpose{A_{21}}})}}\\
& = \MatrixProduct{A_{11}}{(\MatrixProduct{Y}{\Transpose{A_{22}}} + \Transpose{A_{21}})}.
\end{split}
\end{equation}
Thus, as~${S_{3} = S_{1}-A_{22} = \MatrixProduct{(A_{21}-A_{11})}{Y}-A_{22}=-S_{2}-\MatrixProduct{A_{11}}{Y}}$:
\begin{equation}
\label{eq:prf:u1}
\begin{split}
	U_{1} & = P_{1} + P_{5}  = \MatrixProduct{A_{11}}{\Transpose{A_{11}}} + \MatrixProduct{S_{3}}{\Transpose{S_{3}}} \\
	& = \MatrixProduct{A_{11}}{\Transpose{A_{11}}}
	+ \MatrixProduct{(S_{2}+\MatrixProduct{A_{11}}{Y})}{(\Transpose{S_{2}}
	+\MatrixProduct{\Transpose{Y}}{\Transpose{A_{11}}} )} \\
	& = \MatrixProduct{S_{2}}{\Transpose{S_{2}}}+\Transpose{R_{1}}+R_{1}.
\end{split}
\end{equation}
And denote~${R_{2}=\MatrixProduct{A_{21}}{\MatrixProduct{{Y}}{\Transpose{A_{22}}}}}$, so that:
\begin{equation}
\label{eq:prf:s2s2T}
\begin{split}
	\MatrixProduct{S_{2}}{\Transpose{S_{2}}} & = \MatrixProduct{(A_{22} -
	\MatrixProduct{A_{21}}{Y})}{(\Transpose{A_{22}} -
	\MatrixProduct{\Transpose{Y}}{\Transpose{A_{21}}})} \\
	& = \MatrixProduct{A_{22}}{\Transpose{A_{22}}}
		-\MatrixProduct{A_{21}}{\Transpose{A_{21}}}-R_{2}-\Transpose{R_{2}}.
\end{split}
\end{equation}
Furthermore, from Equation~(\ref{eq:prf:r1}):
\begin{equation}
\label{eq:prf:r1p4}
\begin{split}
	  R_{1} + P_{4} & = R_{1} + \MatrixProduct{S_{1}}{\Transpose{S_{2}}}\\
	& = R_{1} + \MatrixProduct{(A_{21} - A_{11})}{\MatrixProduct{Y}{(\Transpose{A_{22}} - \MatrixProduct{\Transpose{Y}}{\Transpose{A_{21}}})}}\\
	& = \MatrixProduct{A_{11}}{(\MatrixProduct{Y}{\Transpose{A_{22}}}
	+\Transpose{A_{21}})} +
	\MatrixProduct{S_{1}}{\Transpose{S_{2}}} \\
	& = \MatrixProduct{A_{21}}{\MatrixProduct{Y}{\Transpose{A_{22}}}}
	+ \MatrixProduct{A_{21}}{\Transpose{A_{21}}}
	  = R_{2} + \MatrixProduct{A_{21}}{\Transpose{A_{21}}}.
\end{split}
\end{equation}
Therefore, from Equations~(\ref{eq:prf:u1}), (\ref{eq:prf:s2s2T}) and~(\ref{eq:prf:r1p4}):
  \begin{equation}\label{eq:prf:u5}\begin{split}
      U_{5} & = U_{1} + P_{4} + \Transpose{P_{4}}
	   = \MatrixProduct{S_{2}}{\Transpose{S_{2}}}+ R_{1} + \Transpose{R_{1}}+P_{4}+\Transpose{P_{4}} \\
	  & = \MatrixProduct{A_{22}}{\Transpose{A_{22}}}
	  +(-1+2)\MatrixProduct{A_{21}}{\Transpose{A_{21}}}  = C_{22}.
    \end{split}\end{equation}
And the last coefficient~$U_{4}$ of the result is obtained from Equations~(\ref{eq:prf:r1p4}) and~(\ref{eq:prf:u5}):
  \begin{equation}\label{eq:prf:u4}\begin{split}
	  &  U_{4} = U_{2} + P_{3} = U_{5} - \Transpose{P_{4}} +P_{3}\\
	  & = U_{2} + \MatrixProduct{A_{22}}{(\Transpose{A_{12}}+\MatrixProduct{\Transpose{Y}}{\Transpose{A_{21}}}-\MatrixProduct{\Transpose{Y}}{\Transpose{A_{11}}}-\Transpose{A_{22}})}\\
	  & = \MatrixProduct{A_{21}}{\Transpose{A_{21}}} - \Transpose{P_{4}}+ \MatrixProduct{A_{22}}{(\Transpose{A_{12}}+\MatrixProduct{\Transpose{Y}}{\Transpose{A_{21}}}-\MatrixProduct{\Transpose{Y}}{\Transpose{A_{11}}})}\\
	  & = \Transpose{R_{1}}-\Transpose{R_{2}}+ \MatrixProduct{A_{22}}{(\Transpose{A_{12}}+\MatrixProduct{\Transpose{Y}}{\Transpose{A_{21}}}-\MatrixProduct{\Transpose{Y}}{\Transpose{A_{11}}})}\\
	  & = \Transpose{R_{1}}+ \MatrixProduct{A_{22}}{(\Transpose{A_{12}}-\MatrixProduct{\Transpose{Y}}{\Transpose{A_{11}}})} \\
	  & =\MatrixProduct{A_{21}}{\Transpose{A_{11}}}+ \MatrixProduct{A_{22}}{\Transpose{A_{12}}} = C_{21}.
    \end{split}
\end{equation}
Finally,~${P_{1}=\MatrixProduct{A_{11}}{\Transpose{A_{11}}}}$,~${P_{2}=\MatrixProduct{A_{12} }{\Transpose{A_{12}}}}$, and~${P_{5}=\MatrixProduct{S_{3}}{\Transpose{S_{3}}}}$ are symmetric by construction.
So are therefore,~${U_{1}=P_{1}+P_{5}}$,~${U_{3}=P_{1}+P_{2}}$ and~${U_{5}=U_{1}+(P_{4}+\Transpose{P_{4}})}$.
\end{proof}

\subsection{Threshold in the theoretical number of operations for dimensions that are a power of two}
Here, we look for a theoretical threshold where our fast symmetric algorithm performs less arithmetic operations than the classical one.
Below that threshold any recursive call should call a classical algorithm for~$\MatrixProduct{A}{\Transpose{A}}$.
But, depending whether padding or static/dynamic peeling is used, this threshold varies.
For powers of two, however, no padding nor peeling occurs and we thus
have a look in this section of the thresholds in this case.

\begin{table}[htbp]
\footnotesize
\begin{center}
\begin{tabular}{|c||c|c||c|c|c|c|c|c|}
\hline
\multicolumn{3}{|c||}{n} & 4 & 8 & 16 & 32 & 64 & 128\\
\hline
\textsc{syrk} &&&70 & 540 & 4216 & 33264 & 264160 & 2105280\\
\hline
& Rec. & SW & & & & & &\\
\hline
Syrk-i&\multirow{3}{*}{1 } & \multirow{3}{*}{0}&\color{blue}\bf 70 &
\color{blue}\bf 540 & \color{blue}\bf 4216 & 33264 & 264160 &
2105280\\
G0-i&&&81 & 554 & \color{orange}\bf 4020 & 30440 & 236496 & 1863584\\
G1-i&&&89 & 586 & \color{red}\bf 4148 & 30952 & 238544 & 1871776\\
G2-i&&&97 & 618 & 4276 &  31464 & 240592 & 1879968 \\
G3-i&&&105 & 650 & 4404 & \color{green}\bf 31976 & 242640 & 1888160\\
\hline
Syrk-i&\multirow{3}{*}{2 } & \multirow{3}{*}{1}&90 & 604 & 4344 &
\color{blue}\bf 32752 & 253920 & 1998784\\
G0-i&&& & 651 & 4190 & \color{orange}\bf 29340 & 217784 & 1674096\\
G1-i&&& & 707 & 4414 & \color{red}\bf 30236 & 221368 & 1688432\\
G2-i&&& & 763 & 4638 & \color{violet}\bf 31132 & 224952 & 1702768 \\
G3-i&&& & 819 & 4862 & 32028 & 228536 & 1717104\\
\hline
Syrk-i&\multirow{3}{*}{3 } & \multirow{3}{*}{2}& & 824 & 5048 & 34160
& \color{blue}\bf 248288 & 1886144\\
G0-i&&& &  & 4929 & 30746 & \color{orange}\bf 210900 & 1546280\\
G1-i&&& &  & 5225 & 31930 & \color{red}\bf 215636 & 1565224\\
G2-i&&& &  & 5521 & 33114 & \color{violet}\bf 220372 & 1584168\\
G3-i&&& &  & 5817 & 34298 & \color{green}\bf 225108 & 1603112\\
\hline
Syrk-i&\multirow{3}{*}{4 } & \multirow{3}{*}{3}& &  & 6908 & 40112 &
260192 & \color{blue}\bf 1838528\\
G0-i&&& &  &  & 36099 & 221390 & \color{orange}\bf 1500540\\
G1-i&&& &  &  & 37499 & 226990 & \color{red}\bf 1522940\\
G2-i&&& &  &  & 38899 & 232590 & \color{violet}\bf 1545340\\
G3-i&&& &  &  & 40299 & 238190 & \color{green}\bf 1567740\\
\hline
\end{tabular}
\caption{Number of arithmetic operations in the multiplication an~$\matrixsize{n}{n}$ matrix by its transpose:
{\color{blue}\bf blue} when Syrk-i (using Strassen-Winograd
with~${i-1}$ recursive levels) is better than other Syrk;
{\color{orange}\bf orange}/{\color{red}\bf red}/{\color{violet}\bf violet}/{\color{green}\bf
  green} when ours (using Strassen-Winograd with~${i-1}$ recursive
levels, and {\color{orange}\bf G0-i} for~$\CC$ / {\color{red}\bf G1-i}
if~$-1$ is a square / {\color{violet}\bf G2-i} or {\color{green}\bf
  G3-i} otherwise, depending whether $-2$ is a square or not) is better than others.}\label{tab:winolevels}
\end{center}
\end{table}

\par
First, from Section~\ref{ssec:skeworthmat}, over~$\CC$, we can choose~${Y=i\,\IdentityMatrixOfSize{n}}$.
Then multiplications by~$i$ are just exchanging the real and imaginary parts. 
In Equation~(\ref{eq:complexity}) this is an extra cost of~${y=0}$ arithmetic operations in usual machine representations of complex numbers.
Overall, for~${y=0}$ (complex case),~${y=1}$ ($-1$ a square in the finite field) or~${y=3}$ (any other finite field), the dominant term of the complexity is anyway unchanged, but there is a small effect on the threshold.
In the following, we denote by~$G0,G1$ and~$G3$ these three variants.

\par
More precisely, we denote by \textsc{syrk} the classical multiplication of a matrix by its transpose.
Then we denote by Syrk-i the algorithm making four recursive calls and two calls to a generic matrix multiplication via Strassen-Winograd's algorithm, the latter with~${i-1}$ recursive calls before calling the classical matrix multiplication.
Finally G1-i (resp.\ G3-i) is our Algorithm~\ref{alg:wishartp} when~$-1$ is a square (resp.\ not a square), with three recursive calls and two calls to Strassen-Winograd's algorithm, the latter with~${i-1}$ recursive calls.
\par
Now, we can see in Table~\ref{tab:winolevels} in which range the thresholds live.
For instance, over a field where~$-1$ is a square, Algorithm~\ref{alg:wishartp} is better for~${{n}\geq{16}}$ with~$1$ recursive level (and thus~$0$ recursive levels for Strassen-Winograd), for~${{n}\geq{32}}$ with~$2$ recursive levels, etc.
Over a field where~$-1$ is not a square, Algorithm~\ref{alg:wishartp} is better for~${{n}\geq{32}}$ with~$1$ recursive level, for~${{n}\geq{64}}$ with~$3$ recursive levels, etc.

\end{document}